\documentclass[12pt]{article}
\usepackage[paper=a4paper,margin=1in]{geometry}
\usepackage{t1enc}      

\usepackage{amsmath}
\usepackage{amsfonts}
\usepackage{amssymb}
\usepackage{amsthm}
\usepackage{graphicx}
\usepackage{mathrsfs}  

\newcommand{\edth} {\mbox{\symbol{'360}}}

\theoremstyle{plain}

\newtheorem*{proposition*}{Proposition}
\newtheorem*{lemma*}{Lemma}

\providecommand{\keywords}[1]
{\small	\textbf{\textit{Keywords:}} #1 }

\numberwithin{equation}{section}

\begin{document}
\bibliographystyle{unsrt}

\title{The `most classical' states of Euclidean invariant elementary quantum 
mechanical systems}
\author{L\'aszl\'o B. Szabados \\
Wigner Research Centre for Physics, \\
H-1525 Budapest 114, P. O. Box 49, EU\\
e-mail: lbszab@rmki.kfki.hu}

\maketitle

\begin{abstract}
Complex techniques of general relativity are used to determine \emph{all} 
the states in two and three dimensional momentum spaces in which the 
equality holds in uncertainty relations for non-commuting basic 
observables of Euclidean invariant elementary quantum mechanical systems, 
even with non-zero intrinsic spin. It is shown that while there is a
1-parameter family of such states for any two components of the angular 
momentum vector operator with any angle between them, such states exist for
a component of the linear and angular momenta \emph{only if} these 
components are orthogonal to each other and hence the problem is reduced to 
the two-dimensional Euclidean invariant case. We also show that the analogous 
states exist for a component of the linear momentum and of the centre-of-mass 
vector \emph{only if} the angle between them is zero or an acute angle. 
\emph{No} such state (represented by a square integrable and differentiable 
wave function) can exist for \emph{any} pair of components of the 
centre-of-mass vector operator. Therefore, the existence of such states 
depends not only on the Lie algebra, but on the choice of its generators 
as well. 
\end{abstract}

\keywords{smallest uncertainty states, Euclidean invariant quantum mechanical 
systems, edth operators}


\section{Introduction}
\label{sec-1}

The so-called canonical coherent states of the quantum mechanical system 
with the Heisenberg algebra as its algebra of basic observables are usually 
interpreted as the `most classical' states of the system. These are precisely 
the states that yield \emph{equality} in the uncertainty relation for the 
canonically conjugate observables. (For a review, see e.g. \cite{Kl,ZFG}). 
The states analogous to these, called the coherent states, have already been 
introduced for systems with more complicated algebra of basic observables, 
e.g. with the Lie algebras $su(2)$, $su(1,1)$ (see e.g. 
\cite{Rad}-\cite{WaSaPa}), with the Lie algebra $e(n)$ of the Euclidean group 
$E(n)$ in $n\geq2$ dimensions (see e.g. \cite{Bi}-\cite{GuMo-Ce}), or even 
with more the general ones (see \cite{Per}). 

In particular, in the case of the Euclidean groups, these states were 
constructed on the circle $S^1$ and on the 2-sphere $S^2$ 
\cite{KoRePa,KoRe1,KoRe02,GaGa,FreGaNo,GuMo-Ce}, and even on general 
$n$-spheres $S^n$ \cite{HaMi}, by constructing $n$ operators that are 
analogous to the annihilation operators of the Heisenberg system. In these 
investigations, the basic observables were the \emph{position} variable and 
the \emph{orbital} angular momentum, but the \emph{spin} part of the total 
angular momentum was \emph{a priori} assumed to be vanishing. 

In our previous paper \cite{Sz1} we determined \emph{all} the states of 
$SU(2)$-invariant quantum mechanical systems, in which the uncertainties of 
two components of angular momentum operators in two \emph{arbitrarily 
chosen} directions, ${\bf J}(\alpha)$ and ${\bf J}(\beta)$, yield equality 
in the uncertainty relation. It turned out that allowing the angle between 
the two components of the angular momentum vector operator to be arbitrary 
\emph{a new quantum mechanical phenomenon emerges}: the expectation values 
of the two components of the angular momentum behave in a symmetric way 
\emph{only if} the classical parameter space of the solutions is extended 
to be a larger space that is homeomorphic to the non-trivial Riemann surface 
known in connection with the complex function $\sqrt{z}$. 

In the present paper, we extend these investigations to $E(3)$-invariant 
(and, as an illustration of the general strategy, the much simpler 
$E(2)$-invariant) \emph{elementary} quantum mechanical systems. (Following 
the work of Newton and Wigner \cite{NeWign}, we call the system 
\emph{elementary} if 
its states belong to the carrier space of some unitary, \emph{irreducible} 
representation of its symmetry group. Here by $E(3)$ we mean the 
\emph{quantum mechanical} Euclidean group, i.e. the semidirect product of 
$SU(2)$ and the translation group $\mathbb{R}^3$, rather than the isometry 
group of the Euclidean 3-space \cite{StWi}.) No time evolution equation for 
the states is used; it is only based on the general \emph{kinematical} 
structure of the theory. Since the components of the linear momentum, as the 
generators of space translations, are already among the elements of 
Lie algebras $e(2)$ and $e(3)$, moreover we think that the momentum space is 
a more natural arena for the formulation of quantum theory than the 
configuration space, we search for these states in the \emph{momentum}, 
rather than in the position representation. (Since in these special cases 
both the momentum and configuration spaces are modeled by some Euclidean 
space, the momentum and position representations are usually considered to 
be physically equivalent. However, \emph{physically}, and hence also 
mathematically, these two spaces are \emph{different}: The former has 
an extra structure as it is a \emph{vector space} with the \emph{physically 
distinguished point} $p^i=0$, the latter is only an \emph{affine space}, in 
which the origin is \emph{not} distinguished and any two of its points are 
equivalent.) No \emph{a priori} restriction is imposed on the unitary, 
irreducible representations, thus, in particular, the intrinsic spin is 
\emph{not} assumed to be zero. 

If ${\bf A}$ and ${\bf B}$ are any two not commuting observables and $\phi$ 
is a normalized state of the system, then, as is well known, the necessary 
and sufficient condition of the equality in the uncertainty relation, 
$\Delta_\phi{\bf A}\Delta_\phi{\bf B}\geq\vert\langle[{\bf A},{\bf B}]\rangle
_\phi\vert/2$, is that $\phi$ is the solution of the eigenvalue equation 
\begin{equation}
({\bf A}-{\rm i}\lambda{\bf B})\phi=(\langle{\bf A}\rangle_\phi-{\rm i}
\lambda\langle{\bf B}\rangle_\phi)\phi  \label{eq:1.1}
\end{equation}
for some non-zero real number $\lambda$. Here e.g. $\langle{\bf A}\rangle
_\phi$ is the expectation value of ${\bf A}$, and the standard deviation in 
the state $\phi$ is given by $\Delta_\phi{\bf A}=\sqrt{\langle{\bf A}^2
\rangle_\phi-(\langle{\bf A}\rangle_\phi)^2}$. It follows from (\ref{eq:1.1}) 
that $\vert\lambda\vert=\Delta_\phi{\bf A}/\Delta_\phi{\bf B}$, and, in what 
follows, we choose $\lambda$ to be positive. We will call these states the 
`most classical' states (with respect to the observables ${\bf A}$ and 
${\bf B}$). (\ref{eq:1.1}) shows why it is almost impossible to find a 
non-trivial state in which the equality in the uncertainty relations would 
hold for all the pairs $({\bf A},{\bf B})$, $({\bf B},{\bf C})$ and $({\bf 
C},{\bf A})$ from the \emph{three} non-commuting observables ${\bf A}$, 
${\bf B}$ and ${\bf C}$: the state $\phi$ would have to be the eigenvector 
of the three non-self-adjoint operators ${\bf A}-{\rm i}\lambda_1{\bf B}$, 
${\bf B}-{\rm i}\lambda_2{\bf C}$ and ${\bf C}-{\rm i}\lambda_3{\bf A}$ at 
the same time for some non-zero reals $\lambda_1$, $\lambda_2$ and $\lambda
_3$. It is equation (\ref{eq:1.1}) that we will solve in the carrier space 
of \emph{unitary, irreducible} representations of $E(2)$ and $E(3)$ when 
the basic observables belong to the Lie algebra $e(2)$ and $e(3)$, 
respectively. 

We find that the existence and the properties of the most classical 
states of elementary systems depend not only on the algebra of 
observables, but also on the actual choice of the pair of observables. 
We show that, in contrast to the $E(2)$-invariant systems, in 
$E(3)$-invariant systems the existence of the most classical states is 
\emph{not} guaranteed for any two non-commuting observables; and even if 
these exist, then the expectation values of the observables could be 
considerably restricted relative to their classical values. 

In particular, 
(1) for the component ${\bf p}(\alpha)$ and ${\bf J}(\beta)$ of the linear 
and total angular momenta in the direction $\alpha^i$ and $\beta^i$, 
respectively, the most classical states exist precisely when the two 
directions are orthogonal to each other, and the expectation value of the 
former is zero and that of the latter is an integer/half-odd-integer times 
of $\hbar$ for systems with integer/half-odd-integer spin; 
(2) for the components ${\bf J}(\alpha)$ and ${\bf J}(\beta)$ of the angular 
momentum the most classical states always exist (see the third paragraph 
above, and \cite{Sz1} for the details); 
(3) for the component ${\bf p}(\alpha)$ and ${\bf C}(\beta)$ of the linear 
momentum and centre-of-mass, respectively, these states exist precisely when 
the angle between the directions $\alpha^i$ and $\beta^i$ is zero or an acute 
angle, and the range of the expectation value of ${\bf p}(\alpha)$ is 
restricted by this angle; 
(4) for the components ${\bf C}(\alpha)$ and ${\bf C}(\beta)$ of the 
centre-of-mass \emph{no such state exists at all} which could be represented 
by any differentiable, square integrable wave function. 
In the cases when these states exist, they form infinite parameter families. 
These are explicitly given in closed form, and they are considerably more 
general than the previously known ones. 

In section \ref{sec-2} we determine the most classical states of the 
$E(2)$-invariant systems and discuss their properties. Then, section 
\ref{sec-3} is devoted to the $E(3)$-invariant systems: first we summarize 
the unitary, irreducible representations of $e(3)$, and then we determine 
the most classical states in the cases mentioned above. The results are 
summarized and discussed in section \ref{sec-4}. In deriving the above 
results, complex techniques of general relativity are used. To make the 
paper (mostly) self-contained, these ideas are summarized in the Appendix. 
For a more detailed discussion of these ideas and techniques, see e.g. 
\cite{PR,HT,EaTod,NP,Goetal,FrSz}. 

Our conventions are mostly those of \cite{PR} (and identical with those of 
\cite{Sz1}). In particular, round brackets around indices denote 
symmetrization. No abstract indices are used, every index is a concrete 
(name, or component) index. 


\section{$E(2)$-invariant elementary systems}
\label{sec-2}

The (abstract) Lie algebra $e(2)$ is spanned by the elements ${\bf p}^1$, 
${\bf p}^2$ and ${\bf J}$ satisfying the commutation relations $[{\bf p}^1,
{\bf p}^2]=0$, $[{\bf p}^1,{\bf J}]={\rm i}\hbar{\bf p}^2$ and $[{\bf p}^2,
{\bf J}]=-{\rm i}\hbar{\bf p}^1$. Thus ${\bf P}^2:=({\bf p}^1)^2+({\bf p}^2)
^2$ is a Casimir operator in its enveloping algebra; and hence an 
\emph{irreducible} representation of $e(2)$ is labeled by the value $P
\geq0$ of ${\bf P}$. The representation space is ${\cal H}:=L_2({\cal S},
P{\rm d}\varphi)$, where ${\cal S}:=\{(p^1,p^2)\in\mathbb{R}^2\,\vert\,(p^1)
^2+(p^2)^2=P^2\,\}$, the circle of radius $P$ in the classical momentum space 
$\mathbb{R}^2$, is analogous to the `mass shell' (or rather the `kinetic 
energy shell'), endowed with the Euclidean scalar product $\delta_{ij}$, $i,j
=1,2$. We lower and raise the small Latin indices freely by the Kronecker 
delta $\delta_{ij}$ and its inverse, respectively. Here, $(p^1,p^2)$ are 
Cartesian coordinates in the momentum space $\mathbb{R}^2$. In the measure 
$P{\rm d}\varphi$ the angle $\varphi\in[0,2\pi)$ is defined by $(p^1,p^2)=
(P\cos\varphi,P\sin\varphi)$. 

An element of the group $E(2)$ is represented by a translation $\xi^i\in
\mathbb{R}^2$ and the angle $\chi\in[0,2\pi)$ of a rotation; and the action 
of such an element $(\xi^i,\chi)$ of $E(2)$ on the function $\phi\in{\cal H}$ 
is $({\bf U}(\xi^i,\chi)\phi)(\varphi):=\exp({\rm i}p_i\xi^i/\hbar)\phi(
\varphi-\chi)$. Clearly, this ${\bf U}$ provides a unitary, irreducible 
representation of $E(2)$, and yields the representation of the abstract Lie 
algebra $e(2)$ as ${\bf p}_i\phi=p_i\phi$ and ${\bf J}\phi={\rm i}\hbar(d\phi/
d\varphi)$ as well. ${\bf p}_i$ are bounded Hermitean operators, and ${\bf J}$ 
is a formally self-adjoint (unbounded) operator on the dense subspace of the 
smooth functions in ${\cal H}$. 

Applying (\ref{eq:1.1}) to the present operators, for the (normalized) state 
$\phi$ we obtain $({\bf p}_1-{\rm i}\lambda_1{\bf J})\phi=(\langle{\bf p}_1
\rangle_\phi-{\rm i}\lambda_1\langle{\bf J}\rangle_\phi)\phi$ and $({\bf p}_2
-{\rm i}\lambda_2{\bf J})\phi=(\langle{\bf p}_2\rangle_\phi-{\rm i}\lambda_2
\langle{\bf J}\rangle_\phi)\phi$. However, in the above representation of the 
operators, these equations yield $P(\cos\varphi/\lambda_1-\sin\varphi/\lambda
_2)\phi=(\langle{\bf p}_1\rangle_\phi/\lambda_1-\langle{\bf p}_2\rangle_\phi/
\lambda_2)\phi$, which could hold true for every $\varphi$ only if $\phi=0$. 
Thus, there is \emph{no} state in which the equality would hold in \emph{both} 
inequalities $\Delta_\phi{\bf p}_1\Delta_\phi{\bf J}\geq\hbar\vert\langle
{\bf p}_2\rangle_\phi\vert/2$ and $\Delta_\phi{\bf p}_2\Delta_\phi{\bf J}
\geq\hbar\vert\langle{\bf p}_1\rangle_\phi\vert/2$. Only one such equation 
can be required to hold (as we already noted following equation 
(\ref{eq:1.1})). 

Thus, let $\alpha_i\in\mathbb{R}^2$ and form the operator ${\bf p}(\alpha):=
\alpha^i{\bf p}_i$. Although the parameters $\alpha_i$ could be arbitrary, it 
is more natural to consider only those of the form $(\alpha_1,\alpha_2)=(\cos
\alpha,\sin\alpha)$, where $\alpha\in[0,2\pi)$, because in this case the 
operator ${\bf p}(\alpha)$ is e.g. ${\bf p}_1$ in the appropriately rotated 
basis. Then $[{\bf p}(\alpha),{\bf J}]=-{\rm i}\hbar\,{\bf p}^\bot(\alpha)$, 
where ${\bf p}^\bot(\alpha):=\alpha_2{\bf p}_1-\alpha_1{\bf p}_2$. Hence 
$\Delta_\phi{\bf p}(\alpha)\Delta_\phi{\bf J}\geq\hbar\vert\langle
{\bf p}^\bot(\alpha)\rangle_\phi\vert/2$, in which, according to (\ref{eq:1.1}), 
the equality holds precisely when $({\bf p}(\alpha)-{\rm i}\lambda{\bf J})
\phi=(\langle{\bf p}(\alpha)\rangle_\phi-{\rm i}\lambda\langle{\bf J}\rangle
_\phi)\phi$. In the above explicit momentum representation, this condition 
takes the form of the differential equation 
\begin{equation}
\lambda\hbar\frac{d\phi}{d\varphi}=\Bigl(\langle{\bf p}(\alpha)\rangle_\phi-
P\cos(\varphi-\alpha)-{\rm i}\lambda\langle{\bf J}\rangle_\phi\Bigr)\phi. 
\label{eq:2.1}
\end{equation}
Its solution is 
\begin{equation*}
\phi=A\exp\Bigl(-\frac{P}{\lambda\hbar}\sin(\varphi-\alpha)+\frac{1}{\lambda
\hbar}\bigl(\langle{\bf p}(\alpha)\rangle_\phi-{\rm i}\lambda\langle{\bf J}
\rangle_\phi\bigr)\varphi\Bigr),
\end{equation*}
where $A$ is a constant. However, this function is periodic in $\varphi$ 
with period $2\pi$ precisely when $\langle{\bf p}(\alpha)\rangle_\phi=0$ and 
$\langle{\bf J}\rangle_\phi=\ell\hbar$, $\ell\in\mathbb{Z}$, and hence 
\begin{equation}
\phi=A\exp\Bigl(-\frac{P}{\lambda\hbar}\sin(\varphi-\alpha)-{\rm i}\ell
\varphi\Bigr). \label{eq:2.2}
\end{equation}
The constant $A$ is fixed by the normalization condition 
\begin{equation}
1=\int^{2\pi}_0\vert\phi\vert^2P\,{\rm d}\varphi=2\pi P\vert A\vert^2 I_0(
-\frac{2P}{\lambda\hbar}), \label{eq:2.3}
\end{equation}
where $I_k$ is the $k$th modified Bessel function of the first kind, $k=
0,1,2,...$; which can be given by 
\begin{equation}
I_k(x)=\frac{1}{2\pi}\int_0^{2\pi}\exp(x\cos(\omega))\cos(k\omega){\rm d}
\omega=\sum_{l=0}^\infty\frac{1}{l!(l+k)!}(\frac{x}{2})^{2l+k}. \label{eq:2.4}
\end{equation}
In particular, $I_k(x)$ is an even function if $k$ is even, and it is odd if 
$k$ is odd. They satisfy the recurrence relations $x(I_k(x)-I_{k+2}(x))=
2(k+1)I_{k+1}(x)$. The asymptotic expansion of $I_k(x)$ for large $x$ is 
$I_k(x)\sim(2\pi x)^{-1/2}\exp(x)(1-(4k^2-1)/(8x)+...)$ (see e.g. \cite{AbSt}, 
pp 375-377). Note also that $I_k=T_k(d/dx)I_0$, where $T_k$ is the 
Chebyshev polynomial of the first kind \cite{We}. In particular, $I_1=
(dI_0/dx)$ and $I_2=2(d^2I_0/dx^2)-I_0$ hold. We will use these formulae 
below, and in the next section, too. 

Then it is straightforward to calculate the expectation values and 
uncertainties: as we should, we do, in fact, recover $\langle{\bf p}(\alpha)
\rangle_\phi=0$ and $\langle{\bf J}\rangle_\phi=\ell\hbar$ from (\ref{eq:2.2}), 
too. In addition, by integration by parts, we obtain that 
\begin{equation*}
\langle\bigl({\bf p}(\alpha)\bigr)^2\rangle_\phi=\lambda\frac{\hbar}{2}\langle
{\bf p}^\bot(\alpha)\rangle_\phi, 
\hskip 20pt
\langle{\bf J}^2\rangle_\phi=\langle{\bf J}\rangle_\phi^2+\frac{1}{\lambda^2}
\langle\bigl({\bf p}(\alpha)\bigr)^2\rangle_\phi; 
\end{equation*}
and hence 
\begin{equation}
\bigl(\Delta_\phi{\bf p}(\alpha)\bigr)^2=\lambda\frac{\hbar}{2}\langle{\bf p}
^\bot(\alpha)\rangle_\phi, 
\hskip 20pt
\bigl(\Delta_\phi{\bf J}\bigr)^2=\frac{1}{\lambda}\frac{\hbar}{2}\langle
{\bf p}^\bot(\alpha)\rangle_\phi. \label{eq:2.5}
\end{equation}
Thus, $\Delta_\phi{\bf p}(\alpha)$ and $\Delta_\phi{\bf J}$ do, indeed, 
saturate the inequality $\Delta_\phi{\bf p}(\alpha)\Delta_\phi{\bf J}\geq
\hbar\vert\langle{\bf p}^\bot(\alpha)\rangle_\phi\vert/2$, and we 
recover $\lambda=\Delta_\phi{\bf p}(\alpha)/\Delta_\phi{\bf J}$, too. Finally, 
to determine $\langle{\bf p}^\bot(\alpha)\rangle_\phi$, we need to calculate 
$\langle{\bf p}_1\rangle_\phi$ and $\langle{\bf p}_2\rangle_\phi$ explicitly. 
Using 
\begin{equation}
\int_0^{2\pi}\exp(x\cos(\omega))\sin(k\omega){\rm d}\omega=0, \hskip 20pt 
k\in\mathbb{Z}, \label{eq:2.6}
\end{equation}
by integration by parts we find that 
\begin{equation}
\langle{\bf p}_1\rangle_\phi=-\sin\alpha\,2\pi\vert A\vert^2P^2\,I_1\bigl(
-\frac{2P}{\lambda\hbar}\bigr)=\sin\alpha\,2\pi\vert A\vert^2P^2\,I_1\bigl(
\frac{2P}{\lambda\hbar}\bigr).  \label{eq:2.7}
\end{equation}
By $0=\langle{\bf p}(\alpha)\rangle_\phi=\cos\alpha\langle{\bf p}_1\rangle
_\phi+\sin\alpha\langle{\bf p}_2\rangle_\phi$, $\langle{\bf p}_2\rangle_\phi$ 
can be determined, too. (Similar calculations yield that $\langle({\bf p}_1)^2
\rangle_\phi=\langle({\bf p}_2)^2\rangle_\phi=P^2/2$, and hence that $\langle
\delta^{ij}{\bf p}_i{\bf p}_j\rangle_\phi=P^2$, as we expected.) Hence, 
\begin{equation*}
\bigl(\Delta_\phi{\bf p}(\alpha)\bigr)^2=\pi\hbar\lambda P^2\vert A\vert^2I_1
\bigl(\frac{2P}{\lambda\hbar}\bigr).
\end{equation*}
Thus, using the normalization condition (\ref{eq:2.3}), finally we obtain 
that 
\begin{equation}
\Delta_\phi{\bf p}(\alpha)\Delta_\phi{\bf J}=\frac{1}{2}P\hbar\frac{I_1\bigl(
\frac{2P}{\lambda\hbar}\bigr)}{I_0\bigl(\frac{2P}{\lambda\hbar}\bigr)}. 
\label{eq:2.8}
\end{equation}
Since $I_0(0)=1$ and $I_1(0)=0$, in the $\lambda\to\infty$ limit the product 
uncertainty $\Delta_\phi{\bf p}(\alpha)\Delta_\phi{\bf J}$ tends to zero. On 
the other hand, since $\lim_{x\to\infty}(I_1(x)/I_0(x))=1$, in the $\lambda\to
0$ limit it tends to a finite value: $\Delta_\phi{\bf p}(\alpha)\Delta_\phi
{\bf J}\to P\hbar/2$. 

Therefore, to summarize, solution (\ref{eq:2.2}) is parameterized by the 
integer $\ell\in\mathbb{Z}$ and the continuous parameter $\lambda\in(0,
\infty)$, the latter being the ratio of the two uncertainties. This ratio is 
unrestricted, although the expectation values themselves are discrete: 
$\langle{\bf p}(\alpha)\rangle_\phi=0$ and $\langle{\bf J}\rangle_\phi=\ell
\hbar$. Although there is no state in the family of states (\ref{eq:2.2}) in 
which the product uncertainty would actually take its minimum, by 
(\ref{eq:2.8}) in the limit $\lambda\to\infty$ we can approximate such an 
ideal state as much as we wish. (Another candidate to be a measure 
of the `overall uncertainty' could be the sum of the square of the 
uncertainties, $\hbar^2(\Delta_\phi{\bf p}(\alpha))^2+P^2(\Delta_\phi{\bf J})
^2$, which gives $(\hbar P/2)(\hbar\lambda^2+P^2/\lambda)(I_1/I_0)$. 
This would yield a \emph{different} notion of the `most classical' states.) 
The role of the parameter $\alpha\in[0,2\pi)$ is only to specify which 
component of the linear momentum is considered, and this can be chosen 
freely. The product uncertainty does not depend on $\alpha$, as it should 
not.


\section{$E(3)$-invariant elementary systems}
\label{sec-3}

\subsection{The  unitary, irreducible representations of $e(3)$}
\label{sub-3.1}

The Lie algebra $e(3)$ is generated by ${\bf p}^i$ and ${\bf J}_i$, $i=1,2,3$, 
satisfying $[{\bf p}^i,{\bf p}^j]=0$, $[{\bf p}_i,{\bf J}_j]={\rm i}\hbar
\varepsilon_{ijk}{\bf p}^k$ and $[{\bf J}_i,{\bf J}_j]={\rm i}\hbar\varepsilon
_{ijk}{\bf J}^k$, where the lowering and raising of the small Latin indices 
are defined by $\delta_{ij}$ and its inverse, and $\varepsilon_{ijk}$ is the 
alternating Levi-Civita symbol. The two Casimir operators are ${\bf P}^2:=
\delta_{ij}{\bf p}^i{\bf p}^j$ and ${\bf W}:={\bf J}_i{\bf p}^i={\bf p}^i
{\bf J}_i$. Thus, in an \emph{irreducible} representation of $e(3)$, the 
Casimir operators are of the form $P^2{\bf I}$ and $w{\bf I}$, respectively, 
with some non-negative $P^2$ and real $w$, where ${\bf I}$ is the identity 
operator. 

The part of ${\bf J}_i$ that ${\bf W}$ does not fix is ${\bf C}_i:=-
\varepsilon_{ikl}{\bf J}^k{\bf p}^l+{\rm i}\hbar{\bf p}_i$. (Since in itself 
$-\varepsilon_{ikl}{\bf J}^k{\bf p}^l$ is \emph{not} self-adjoint even if 
${\bf p}^i$ and ${\bf J}_i$ are, the term ${\rm i}\hbar{\bf p}_i$ is needed 
to ensure the self-adjointness of ${\bf C}_i$.) Then 
\begin{equation}
[{\bf p}_i,{\bf C}_k]=-{\rm i}\hbar\bigl(\delta_{ik}{\bf P}^2-{\bf p}_i
{\bf p}_k\bigr), \hskip 20pt
[{\bf C}_i,{\bf C}_k]=-{\rm i}\hbar\Bigl(\varepsilon_{ikl}{\bf p}^l{\bf W}
+{\bf C}_i{\bf p}_k-{\bf C}_k{\bf p}_i\Bigr). \label{eq:3.1.1}
\end{equation}
In an irreducible representation, the part ${\bf C}_i$ of ${\bf J}_i$ can be 
interpreted as $P^2$-times the centre-of-mass vector operator. ${\bf W}$ 
represents the intrinsic spin, or rather the helicity, with respect to the 
3-momentum, while ${\bf L}^i:=\varepsilon^{ijk}{\bf C}_j{\bf p}_k/P^2={\bf J}
^i-w{\bf p}^i/P^2$ is the orbital angular momentum. ${\bf J}_i{\bf J}^i$ is 
a Casimir operator only of the $su(2)$ subalgebra, and $P^2{\bf J}_i{\bf J}
^i=(w^2-\hbar^2P^2){\bf I}+{\bf C}_i{\bf C}^i$ holds. 

The unitary, irreducible representations of $E(3)$ and $e(3)$ in their 
traditional, purely algebraic bra-ket formalism are known and are summarized 
e.g. in \cite{Tung,KoRe1,SiJa}. Nevertheless, they can also be derived and 
presented in a geometric form using Weyl spinor fields on 2-spheres of the 
momentum space, analogously to those on the mass-shell for the Poincar\'e 
group (see \cite{StWi}). However, a special feature of the spinor fields 
in the latter formalism, viz. their N-type character (see Appendix 
\ref{sub-A.2}, or for a more detailed discussion, \cite{PR,HT}), is not 
manifest. Moreover, and more importantly, the naive, apparently `obvious' 
decomposition of the total angular momentum operator into its spin and 
orbital parts is \emph{not} well defined. This drawback is cured by the use 
of the space of the square integrable cross section of the complex line 
bundles ${\cal O}(-2s)$ as the representation space. These two forms of 
the representations are motivated by the spinorial ideas and techniques of 
general relativity \cite{EaTod,PR,HT}; and, in the case of 
$SU(2)$-invariant systems considered in \cite{Sz1}, their use has already 
made considerably easier to solve the resulting partial differential 
equations in closed form. In Appendix \ref{sub-A.1}, we summarize the complex 
and spinorial tools that we use, and Appendix \ref{sub-A.2} is a summary of 
the key results on the unitary, irreducible representations of $E(3)$ in 
these two forms. These appendices make the present paper essentially 
self-contained. For a more detailed discussion of these ideas, see 
\cite{EaTod,PR,HT}, or, for a concise summary of the necessary background, 
see Appendix A.1 of \cite{Sz1}. 

Thus, the concrete realization of the unitary, irreducible representations 
of $e(3)$ that we use here is based on the complex line bundle ${\cal O}
(-2s)$ over the 2-sphere ${\cal S}$ of radius $P$, where $2s\in
\mathbb{Z}$ is fixed. This 2-sphere, ${\cal S}:=\{p^i\in\mathbb{R}^3\,\vert\,
p^ip^j\delta_{ij}=P^2\,\}$, is thought of as the `kinetic energy shell' in the 
classical momentum 3-space $(\mathbb{R}^3,\delta_{ij})$, in which $p^i$ are 
Cartesian coordinates. The carrier space of the representation is the Hilbert 
space ${\cal H}_s=L_2({\cal S},{\rm d}{\cal S})$ of the square-integrable 
cross sections of the line bundle, where the measure ${\rm d}{\cal S}$ is the 
natural metric area element on ${\cal S}$. The integral of the first Chern 
class of ${\cal O}(-2s)$ is $2s$, which is a topological invariant, and hence 
$s$ characterizes the global non-triviality (or `twist') of the bundle (see 
e.g. \cite{FrSz}). 

In this representation, the $e(3)$ Casimir operators have the form ${\bf P}
^2\phi=P^2\phi$ and ${\bf W}\phi=\hbar Ps\phi$ (see Appendix \ref{sub-A.2}), 
and the basic observables and the centre-of-mass operator act on 
$\phi$ as 
\begin{eqnarray}
&{}&{\bf p}^i\phi=p^i\phi, \label{eq:3.1.2a}\\
&{}&{\bf J}^i\phi=s\hbar\frac{p^i}{P}\phi+{\bf L}^i\phi=s\hbar\frac{p^i}{P}
  \phi+P\hbar\bigl(m^i\,{\edth}'\phi-\bar m^i\,{\edth}\phi\bigr), 
  \label{eq:3.1.2b}\\
&{}&{\bf C}_i\phi={\rm i}\hbar\bigl(P^2m_i\,{\edth}'\phi+P^2\bar m_i\,{\edth}
  \phi-p_i\phi\bigr). \label{eq:3.1.2c}
\end{eqnarray}
In particular, the spin of the system is encoded into the twist $s$ of the 
line bundle ${\cal O}(-2s)$. For the definition of the complex null tangents 
$m^i$ and $\bar m^i$ of ${\cal S}$, the operators ${\edth}$ and ${\edth}'$ 
and the related concepts, and the derivation of 
(\ref{eq:3.1.2a})-(\ref{eq:3.1.2c}), see the appendices 
\ref{sub-A.1}-\ref{sub-A.2}. In this representation, the $su(2)$ Casimir 
operator is ${\bf J}_i{\bf J}^i\phi=s^2\hbar^2\phi-P^2\hbar^2({\edth}{\edth}'
+{\edth}'{\edth})\phi$. The square of the centre-of-mass vector operator 
deviates from this only in a term proportional with the identity operator: 
${\bf C}_i{\bf C}^i\phi=P^2{\bf J}_i{\bf J}^i\phi+P^2\hbar^2(1-s^2)\phi{\bf I}$. 

\subsection{The most classical states for the $({\bf p}^i,{\bf J}_i)$ 
system}
\label{sub-3.2}

If ${\bf p}^i$ and ${\bf J}_i$ are considered to be the basic observables in 
$e(3)$, then there are two kinds of non-commuting observables: one linear and 
one angular momentum component, and two different angular momentum components. 
We discuss these cases separately. 


\subsubsection{The most classical states for the observables $({\bf p}
(\alpha),{\bf J}(\beta))$}
\label{sub-3.2.1}

For any $\alpha_i,\beta^i\in\mathbb{R}^3$ satisfying $\alpha_i\alpha_j\delta
^{ij}=\beta^i\beta^j\delta_{ij}=1$, we form the operators ${\bf p}(\alpha):=
\alpha_i{\bf p}^i$ and ${\bf J}(\beta):=\beta^i{\bf J}_i$, i.e. the components 
of the basic observables determined by the directions  $\alpha^i$ and $\beta
^i$. However, without loss of generality, we may assume that e.g. $\beta^1=
\beta^2=0$ and $\beta^3=1$, because by an appropriate rotation of the 
Cartesian coordinate system this can always be achieved. Using the 
commutators of the basic observables, in any normalized state $\phi$ we 
obtain $[{\bf p}(\alpha),{\bf J}(\beta)]={\rm i}\hbar\,\alpha^i\beta^j
\varepsilon_{ijk}{\bf p}^k$. Hence, we assume that $\alpha_3\not=\pm1$, 
because otherwise $\alpha^i=\pm\beta^i$ would be allowed, and for these 
${\bf p}(\alpha)$ and ${\bf J}(\beta)$ would commute. Then 
\begin{equation}
\Delta_\phi{\bf p}(\alpha)\Delta_\phi{\bf J}(\beta)\geq\frac{1}{2}\vert\langle
[{\bf p}(\alpha),{\bf J}(\beta)]\rangle_\phi\vert=\frac{\hbar}{2}\vert\alpha^i
\beta^j\varepsilon_{ijk}\langle{\bf p}^k\rangle_\phi\vert, \label{eq:3.2.1}
\end{equation}
in which, by (\ref{eq:1.1}), the condition of the equality is 
\begin{equation}
\bigl({\bf p}(\alpha)-{\rm i}\lambda{\bf J}(\beta)\bigr)\phi=\Bigl(\langle
{\bf p}(\alpha)\rangle_\phi-{\rm i}\lambda\langle{\bf J}(\beta)\rangle_\phi
\Bigr)\phi \label{eq:3.2.2}
\end{equation}
for some \emph{positive} $\lambda$. By (\ref{eq:3.1.2a})-(\ref{eq:3.1.2b}), 
in the unitary, irreducible representation labeled by $P$ and $s$, this 
condition takes the form 
\begin{equation}
-{\rm i}\lambda P\hbar\beta^i\bigl(m_i{\edth}'\phi-\bar m_i{\edth}\phi\bigr)
+\bigl(\alpha_ip^i-{\rm i}\lambda\hbar s\frac{\beta_ip^i}{P}\bigr)\phi=
\bigl(\langle{\bf p}(\alpha)\rangle_\phi-{\rm i}\lambda\langle{\bf J}(\beta)
\rangle_\phi\bigr)\phi. \label{eq:3.2.3}
\end{equation}
Using $\beta^i=(0,0,1)$ and the explicit form of the operators ${\edth}$ and 
${\edth}'$ and of the Cartesian components of the vectors $m_i$ and $\bar 
m_i$ expressed by the polar coordinates $(\theta,\varphi)$ (see Appendix 
\ref{sub-A.1}), we obtain the differential equation 
\begin{equation}
\frac{\partial\ln\phi}{\partial\varphi}=\frac{P}{\lambda\hbar}\sqrt{1-
\alpha^2_3}\sin\theta\cos(\varphi-\alpha)+\Bigl(\frac{P}{\lambda\hbar}
\alpha_3\cos\theta-\frac{1}{\lambda\hbar}\langle{\bf p}(\alpha)\rangle_\phi
\Bigr)+{\rm i}\bigl(\frac{1}{\hbar}\langle{\bf J}_3\rangle_\phi+s\bigr),
\label{eq:3.2.4}
\end{equation}
where we parameterized $(\alpha_1,\alpha_2)$ as $\sqrt{1-\alpha^2_3}(\cos
\alpha,\sin\alpha)$. Its solution is 
\begin{equation*}
\phi=A\exp\Bigl\{\frac{P}{\lambda\hbar}\sqrt{1-\alpha^2_3}\sin\theta\sin(
\varphi-\alpha)+\bigl(\frac{P}{\lambda\hbar}\alpha_3\cos\theta-\frac{1}{
\lambda\hbar}\langle{\bf p}(\alpha)\rangle_\phi\bigr)\varphi+{\rm i}\bigl(
\frac{1}{\hbar}\langle{\bf J}_3\rangle_\phi+s\bigr)\varphi\Bigr\},
\end{equation*}
where $A=A(\theta)$ is an arbitrary complex \emph{function} of $\theta$. 
This $\phi$ is periodic in $\varphi$ with $2\pi$ period precisely when the 
second term is vanishing for any value of $\theta$, and, in the third term, 
the coefficient of $\varphi$ between the round brackets is an integer, say 
$\ell$. This holds precisely when 
\begin{equation}
\alpha_3=0, \hskip 25pt \langle{\bf p}(\alpha)\rangle_\phi=0, \hskip 25pt 
\langle{\bf J}_3\rangle_\phi=(\ell-s)\hbar, \hskip 20pt \ell\in\mathbb{Z}.
\label{eq:3.2.5}
\end{equation}
Then the solution of (\ref{eq:3.2.4}) takes the form 
\begin{equation}
\phi=A\exp\Bigl(\frac{P\sin\theta}{\lambda\hbar}\sin\bigl(\varphi-\alpha
\bigr)+{\rm i}\ell\varphi\Bigr). \label{eq:3.2.6}
\end{equation}
Although its structure for any given $\theta$ is just that of solution 
(\ref{eq:2.2}) for the $E(2)$-invariant system, it depends on one arbitrary 
complex \emph{function}, $A=A(\theta)$, restricted only by the normalization 
condition 
\begin{equation}
1=\int_{\cal S}\vert\phi\vert^2{\rm d}{\cal S}=2\pi P^2\int_0^{\pi}\vert A\vert
^2I_0\sin\theta{\rm d}\theta. \label{eq:3.2.7}
\end{equation}
Here $I_0$ is the 0th modified Bessel function of the first kind, whose 
argument now is $2P\sin\theta/(\lambda\hbar)$. Using (\ref{eq:2.3}), 
(\ref{eq:2.6}), (\ref{eq:3.2.6}) and this normalization condition, it is 
straightforward to calculate $\langle{\bf p}_i\rangle_\phi$ and $\langle
{\bf J}_3\rangle_\phi$, and we do, in fact, recover $\langle{\bf p}(\alpha)
\rangle_\phi=0$ and $\langle{\bf J}_3\rangle_\phi=(\ell-s)\hbar$ (and, as we 
expect, $\langle\delta^{ij}{\bf p}_i{\bf p}_j\rangle_\phi=P^2$, too); and, in 
addition, we find that 
\begin{equation*}
\sin\alpha\langle{\bf p}_1\rangle_\phi-\cos\alpha\langle{\bf p}_2\rangle_\phi=
-2\pi P^3\int_0^{\pi}\vert A\vert^2I_1\sin^2\theta{\rm d}\theta.
\end{equation*}
Using $\alpha_3=0$, $\langle({\bf J}_3)^2\rangle_\phi=\langle{\bf J}_3\phi,
{\bf J}_3\phi\rangle$ and the recurrence relation $x(I_0(x)-I_2(x))=2I_1(x)$, 
similar calculations yield 
\begin{equation}
\bigl(\Delta_\phi{\bf p}(\alpha)\bigr)^2=\lambda^2\bigl(\Delta_\phi{\bf J}_3
\bigr)^2=\pi\lambda\hbar P^3\int_0^{\pi}\vert A\vert^2I_1\sin^2\theta{\rm d}
\theta. \label{eq:3.2.8}
\end{equation}
Therefore, the product uncertainty is $\Delta_\phi{\bf p}(\alpha)\Delta
_\phi{\bf J}_3=(\Delta_\phi{\bf p}(\alpha))^2/\lambda$; and the equality does, 
in fact, hold in (\ref{eq:3.2.1}). 

Next we clarify the $\lambda\to\infty$ and $\lambda\to0$ limits of 
uncertainties. Using the recurrence relation and re-expressing $I_1$ in terms 
of $I_0$ and $I_2$, and recalling that the Bessel functions $I_k(x)$ are 
non-negative for $x\geq0$, by the normalization condition (\ref{eq:3.2.7}) 
equation (\ref{eq:3.2.8}) gives that 
\begin{equation*}
\bigl(\Delta_\phi{\bf p}(\alpha)\bigr)^2=\pi P^4\int^{\pi}_0\vert A\vert^2
\bigl(I_0-I_2\bigr)\sin^3\theta\,{\rm d}\theta\leq\pi P^4\int^{\pi}_0\vert A
\vert^2I_0\sin\theta\,{\rm d}\theta=\frac{1}{2}P^2; 
\end{equation*}
i.e. $\Delta_\phi{\bf p}(\alpha)$ is \emph{bounded}. Hence, by $\Delta_\phi
{\bf J}(\beta)=\Delta_\phi{\bf p}(\alpha)/\lambda$ both $\Delta_\phi{\bf J}
(\beta)$ and the product uncertainty $\Delta_\phi{\bf p}(\alpha)\Delta_\phi
{\bf J}(\beta)$ tend to zero in the $\lambda\to\infty$ limit. To see the 
$\lambda\to0$ limit, first observe that, by (\ref{eq:3.2.7}), the 
asymptotic behaviour of $\vert A\vert^2$ is determined by that of $I_0$: 
with the notation $x_0:=2P/(\lambda\hbar)$, for large $x_0$ it is 
\begin{equation*}
\vert A\vert^2\sim\frac{1}{2\pi^2P^2}\exp\bigl(-x_0\sin\theta\bigr)\Bigl(
\frac{8x_0\sqrt{2\pi x_0\sin\theta}}{1+8x_0\sin\theta}+\cdots\Bigr).
\end{equation*}
Taking into account the asymptotic form of $I_1$ for large $x$, we find that 
$(\Delta_\phi{\bf p}(\alpha))^2=P\hbar\lambda/\pi+O(\lambda^2)$, and hence 
$(\Delta_\phi{\bf J}(\beta))^2$ diverges, while the product uncertainty tends 
to the \emph{finite} value $\hbar P/\pi$ in the $\lambda\to0$ limit. 

Therefore, \emph{the most classical states for the pair $({\bf p}(\alpha),
{\bf J}(\beta))$ of observables can exist only if the two directions are 
orthogonal to each other, $\alpha_i\beta^i=\alpha_3=0$}. 
However, with $\alpha_3=0$ the operators ${\bf p}(\alpha)$ and ${\bf J}
(\beta)$ generate an $e(2)$ sub-Lie algebra in $e(3)$: with the definition 
${\bf p}^\bot(\alpha):=\alpha^i\beta^j\varepsilon_{ijk}{\bf p}^k$ one has, in 
fact, that $[{\bf p}(\alpha),{\bf J}(\beta)]={\rm i}\hbar\,{\bf p}^\bot
(\alpha)$, $[{\bf p}^\bot(\alpha),{\bf J}(\beta)]=-{\rm i}\hbar\,{\bf p}
(\alpha)$ and $[{\bf p}(\alpha),{\bf p}^\bot(\alpha)]=0$. Thus, it is not a 
surprise that the states (\ref{eq:3.2.6}) are similar to the ones obtained 
for the $E(2)$-invariant systems and depend on an integer $\ell$ and the 
parameter $\lambda$ (although they depend on \emph{one almost free function 
of $\theta$}, too, which is restricted only by the normalization condition). 
The range of $\lambda$ is unrestricted, which can be any positive number 
though the expectation values $\langle{\bf p}(\alpha)\rangle_\phi$ and 
$\langle{\bf J}(\beta)\rangle_\phi$ could take only discrete values. 
$\langle{\bf J}(\beta)\rangle_\phi/\hbar$ is half-odd-integer iff $s$ is. 
The product uncertainty $\Delta_\phi{\bf p}(\alpha)\Delta_\phi{\bf J}
(\beta)$ can be made as small as we wish by choosing $\lambda$ large enough, 
i.e. when we approach an eigenstate of ${\bf J}(\beta)$; but, in the 
$\lambda\to0$ limit (i.e. when we approach an eigenstate of ${\bf p}
(\alpha)$) the product uncertainty cannot be made smaller than the finite 
value $\hbar P/\pi$. 


\subsubsection{A summary of the most classical states for the 
observables $({\bf J}(\alpha),{\bf J}(\beta))$}
\label{sub-3.2.2}

The case of two angular momentum components, ${\bf J}(\alpha):=\alpha^i
{\bf J}_i$ and ${\bf J}(\beta):=\beta^i{\bf J}_i$ for any two unit vectors 
$\alpha^i$ and $\beta^i$ for which $\alpha_i\beta^i\not=\pm1$, has already 
been clarified in \cite{Sz1}: both the expectation values and the standard 
deviations and the corresponding wave functions have been determined. 
Thus, for the sake of completeness, we only summarize the key results and 
present those details that we need in subsection \ref{sub-3.3.2}. 

Condition (\ref{eq:1.1}) of the equality in the uncertainty relation 
for ${\bf J}(\alpha)$ and ${\bf J}(\beta)$ is 
\begin{equation}
\bigl({\bf J}(\alpha)-{\rm i}\lambda{\bf J}(\beta)\bigr)\phi=\Bigl(\langle
{\bf J}(\alpha)\rangle_\phi-{\rm i}\lambda\langle{\bf J}(\beta)\rangle_\phi
\Bigr)\phi=:\hbar C\phi \label{eq:3.2.9}
\end{equation}
for some $\lambda>0$. By (\ref{eq:3.1.2b}) in the irreducible representation 
labeled by $P$ and $s$, this takes the explicit form  
\begin{equation}
P\bigl(\alpha^i-{\rm i}\lambda\beta^i\bigr)\Bigl(m_i{\edth}'\phi-\bar m_i
{\edth}\phi+s\frac{p_i}{P^2}\phi\Bigr)=C\phi. \label{eq:3.2.9a}
\end{equation}
Using spinorial techniques of general relativity, and in particular the 
principal spinors of the spinor form of the complex spatial vector $\alpha_i
-{\rm i}\lambda\beta_i$, we have already determined the complex eigenvalue 
$C$. It is given by 
\begin{equation}
C=m\sqrt{1-\lambda^2-2{\rm i}\lambda\alpha_3}, \label{eq:3.2.10}
\end{equation}
where $m=-j,-j+1,...,j$, and $j$ is such that $j=\vert s\vert,\vert s
\vert+1,...$. Hence, by (\ref{eq:3.2.10}), the expectation values are 
\begin{eqnarray}
\langle{\bf J}(\alpha)\rangle_\phi\!\!\!\!&=\!\!\!\!&m\frac{\hbar}
  {\sqrt{2}}\sqrt{1-\lambda^2+\sqrt{(1-\lambda^2)^2+4\lambda^2\alpha^2_3}}, 
  \label{eq:3.2.11a}\\
\lambda\langle{\bf J}(\beta)\rangle_\phi\!\!\!\!&=\!\!\!\!&{\rm sign}
  (\alpha_3)\,m\frac{\hbar}{\sqrt{2}}\sqrt{\lambda^2-1+\sqrt{(1-\lambda^2)
  ^2+4\lambda^2\alpha^2_3}} \label{eq:3.2.11b}
\end{eqnarray}
if $\alpha_3\not=0$; these are $m\hbar\sqrt{1-\lambda^2}$ and 0, respectively, 
if $\alpha_3=0$ and $\lambda<1$; while these are 0 and $m\hbar\sqrt{\lambda^2
-1}$, respectively, if $\alpha_3=0$ and $\lambda>1$. Thus, these expectation 
values depend on the discrete `quantum number' $m$, and the two continuous 
parameters $\alpha_3$ and $\lambda$. The expectation values are zero if 
$\alpha_3=0$ and $\lambda=1$. 

However, for given non-zero $m$, $\langle{\bf J}(\alpha)\rangle_\phi$ is a 
continuous function on the parameter space ${\cal P}:=\{(\alpha_3,\lambda)\,
\vert\,\alpha_3\in(-1,1),\,\lambda\in(0,\infty)\}$, but $\langle{\bf J}
(\beta)\rangle_\phi$ is \emph{not} continuous at $\alpha_3=0$ for $\lambda
>1$. Since physically neither angular momentum component is distinguished 
over the other, in \cite{Sz1} we concluded that the two expectation values 
should behave in a symmetric way. In fact, the standard deviations showed 
this symmetry: they turned out to be continuous on ${\cal P}$, and the 
product uncertainty $\Delta_\phi{\bf J}(\alpha)\Delta_\phi{\bf J}(\beta)$ 
tended to zero \emph{both} in the $\lambda\to0$ and $\lambda\to\infty$ 
limits. Hence, the proper parameter space ${\cal R}$ on which the 
expectation values (and the states also) should depend is homeomorphic to 
the Riemann surface known in connection with the function $\sqrt{z}$ in 
complex analysis. This ${\cal R}$ is obtained from ${\cal P}$ by cutting 
it along the $\lambda\geq1$ segment of the $\alpha_3=0$ axis, and identifying 
the resulting edges with the corresponding opposite edges in a second copy 
of ${\cal P}$ that has been cut in the same way. The branch point of 
${\cal R}$ corresponds to the point $\alpha_3=0$ and $\lambda=1$ of 
${\cal P}$, and we refer to this as the \emph{exceptional case}, and to any 
other as the \emph{generic case}. Then, extending the expectation values to 
the second copy of ${\cal P}$ in ${\cal R}$ to be their own negative, the 
resulting expectation values do, in fact, become differentiable on ${\cal R}$. 

In \cite{Sz1}, we determined the eigenfunctions of the eigenvalue equation 
(\ref{eq:3.2.9a}) in closed form, too. Since it is precisely this technique 
that we will use in subsection \ref{sub-3.3.2}, we summarize its key points. 
Parameterizing the components of $\alpha_i$ by $\alpha_3$ via $(\alpha_1,
\alpha_2)=\sqrt{1-\alpha^2_3}(\cos\alpha,\sin\alpha)$, where $\alpha\in[0,
2\pi)$, and writing (\ref{eq:3.2.9a}) in the complex stereographic 
coordinates $(\zeta,\bar\zeta)$, we find that it is more convenient to use 
the coordinates $(\xi,\bar\xi)$ defined by $\xi:=\exp[-{\rm i}\alpha]\zeta$. 
In these coordinates, (\ref{eq:3.2.9a}) takes the form 
\begin{equation}
X(\ln\phi)=C-\frac{1}{2}s\Bigl(\sqrt{1-\alpha^2_3}\xi-\alpha_3+{\rm i}\lambda
\Bigr)-\frac{1}{2}s\Bigl(\sqrt{1-\alpha^2_3}\bar\xi-\alpha_3+{\rm i}\lambda
\Bigr), \label{eq:3.2.12}
\end{equation}
where the \emph{complex} vector field $X$ on the left is defined by 
\begin{equation}
X:=\Bigl(\frac{1}{2}\sqrt{1-\alpha_3^2}(1-\xi^2)+(\alpha_3-{\rm i}\lambda)
\xi\Bigr)\frac{\partial}{\partial\xi}-\Bigl(\frac{1}{2}\sqrt{1-\alpha_3^2}
(1-\bar\xi^2)+(\alpha_3-{\rm i}\lambda)\bar\xi\Bigr)\frac{\partial}{\partial
\bar\xi}. \label{eq:3.2.13}
\end{equation}
Introducing the notation 
\begin{equation}
\xi_{\pm}:=\frac{\alpha_3-{\rm i}\lambda\pm\sqrt{1-\lambda^2-2{\rm i}\lambda
\alpha_3}}{\sqrt{1-\alpha^2_3}}, \label{eq:3.2.14}
\end{equation}
in the \emph{generic case} the general \emph{local} solution of 
(\ref{eq:3.2.12}) is 
\begin{equation}
\phi=\phi_0\Bigl(\frac{(\xi-\xi_-)(\bar\xi-\xi_+)}{(\xi-\xi_+)(\bar\xi-\xi_-)}
\Bigr)^{m/2}\Bigl(\frac{(\xi-\xi_+)(\xi-\xi_-)}{(\bar\xi-\xi_+)(\bar\xi-\xi_-)}
\Bigr)^{s/2}, \label{eq:3.2.15}
\end{equation}
where $\phi_0$ is an arbitrary smooth complex function of 
\begin{equation}
w:=\frac{(\xi-\xi_+)(\bar\xi-\xi_+)}{(\xi-\xi_-)(\bar\xi-\xi_-)}. 
\label{eq:3.2.16}
\end{equation}
In particular, $\phi_0$ could be 
\begin{equation*}
\frac{(\xi-\xi_+)^a(\bar\xi-\xi_+)^a(\xi-\xi_-)^b(\bar\xi-\xi_-)^b}
{(1+\xi\bar\xi)^{a+b}}=w^a\Bigl(\frac{\xi_+-\xi_-}{\xi_+-\xi_- w}\Bigr)^{a+b}
\end{equation*}
with \emph{arbitrary} real $a$ and $b$. However, with such a general $\phi_0$ 
(\ref{eq:3.2.15}) is only a \emph{local} solution of (\ref{eq:3.2.9a}): we 
still have to ensure that $\phi$ be well defined even on small circles 
surrounding the poles, and be square integrable, too. These requirements 
restrict the structure of $\phi_0$, and, in particular, $a$ and $b$ are 
restricted to be only non-negative integer or half-odd-integer for which 
$a+b=j$. 

If $\lambda=1$ and $\alpha_3=0$ (exceptional case), then $\xi_\pm=-{\rm i}$ 
and the general \emph{local} solution of equation (\ref{eq:3.2.9a}) is 
\begin{equation}
\phi=\phi_0\bigl(\frac{{\rm i}+\xi}{{\rm i}+\bar\xi}\bigr)^s, 
\label{eq:3.2.17}
\end{equation}
where we have used that in the exceptional case the eigenvalue $C$ is zero, 
and $\phi_0$ is an arbitrary smooth complex function of 
\begin{equation}
v:=\frac{1}{{\rm i}+\xi}+\frac{1}{{\rm i}+\bar\xi}. \label{eq:3.2.18}
\end{equation}
In particular, $\phi_0$ could be 
\begin{equation*}
\Bigl(\frac{(\xi+{\rm i})(\bar\xi+{\rm i})}{(1+\xi\bar\xi)}\Bigr)^a
=\frac{1}{(1-{\rm i}v)^a}
\end{equation*}
with \emph{arbitrary} real $a$. However, by the requirement that the 
corresponding $\phi$ be well defined and square integrable, this $a$ is 
restricted to be a non-negative integer or half-odd-integer $j$ for which 
$j=\vert s\vert+n$ holds for some $n=0,1,2,...$. 

To summarize, for \emph{any} pair of angular momentum components in any 
irreducible representation of the $su(2)$ subalgebra labeled by $s$ and 
$j$, there is a 1-parameter family of most classical states which, in 
addition, depend on the discrete `quantum number' $m=-j,-j+1,...,j$, too. 
There is \emph{no} restriction on $\alpha_3$, i.e. on the angle between 
these two components. 


\subsection{The most classical states for the $({\bf p}^i,{\bf C}_i)$ 
system}
\label{sub-3.3}

Since the total angular momentum ${\bf J}_i$ can equivalently be given by 
the Casimir operator ${\bf W}$ and the centre-of-mass vector operator ${\bf 
C}_i$, it is natural to ask for the most classical states with respect to 
the pairs of observables $({\bf p}(\alpha),{\bf C}(\beta))$ and $({\bf C}
(\alpha),{\bf C}(\beta))$. We consider these cases separately.

\subsubsection{The most classical states for the observables $({\bf p}
(\alpha),{\bf C}(\beta))$}
\label{sub-3.3.1}

For any $\alpha^i,\beta^i\in\mathbb{R}^3$ satisfying $\alpha^i\alpha^j\delta
_{ij}=\beta^i\beta^j\delta_{ij}=1$, we form ${\bf p}(\alpha):=\alpha^i{\bf p}
_i$ and ${\bf C}(\beta):=\beta^i{\bf C}_i$. For these 
\begin{equation}
\Delta_\phi{\bf p}(\alpha)\,\Delta_\phi{\bf C}(\beta)\geq\frac{1}{2}\vert
\langle[{\bf p}(\alpha),{\bf C}(\beta)]\rangle_\phi\vert=\frac{\hbar}{2}
\vert P^2\alpha_i\beta^i-\langle{\bf p}(\alpha){\bf p}(\beta)\rangle_\phi
\vert, \label{eq:3.3.1}
\end{equation}
in which the equality holds precisely when, for some \emph{positive} 
$\lambda$, 
\begin{equation*}
\bigl({\bf p}(\alpha)-{\rm i}\lambda{\bf C}(\beta)\bigr)\phi=\Bigl(\langle
{\bf p}(\alpha)\rangle_\phi-{\rm i}\lambda\langle{\bf C}(\beta)\rangle_\phi
\Bigr)\phi 
\end{equation*}
holds. Using (\ref{eq:3.1.2a}) and (\ref{eq:3.1.2c}), this condition is the 
differential equation 
\begin{equation}
P^2\beta^i\bigl(m_i{\edth}'\phi+\bar m_i{\edth}\phi\bigr)+\bigl(
\frac{\alpha_i}{\lambda\hbar}-\beta_i\bigr)p^i\phi=\frac{1}{\lambda\hbar}
\Bigl(\langle{\bf p}(\alpha)\rangle_\phi-{\rm i}\lambda\langle{\bf C}(\beta)
\rangle_\phi\Bigr)\phi. \label{eq:3.3.2}
\end{equation}
As earlier, without loss of generality, we assume that $\beta^i=(0,0,1)$ and 
we use the parameterization $(\alpha_1,\alpha_2)=\sqrt{1-\alpha^2_3}(\cos
\alpha,\sin\alpha)$. Then, in the polar coordinates $(\theta,\varphi)$, this 
equation takes the form 
\begin{equation}
\frac{\partial\ln\phi}{\partial\theta}=-\frac{1}{\sin\theta}\frac{\langle
{\bf p}(\alpha)\rangle_\phi-{\rm i}\lambda\langle{\bf C}_3\rangle_\phi}
{P\lambda\hbar}+\frac{\alpha_3-\lambda\hbar}{\lambda\hbar}\cot\theta+
\frac{1}{\lambda\hbar}\sqrt{1-\alpha^2_3}\cos(\varphi-\alpha). 
\label{eq:3.3.3}
\end{equation}
Its solution is 
\begin{equation}
\phi=A\exp\Bigl(\frac{\theta}{\lambda\hbar}\sqrt{1-\alpha^2_3}\cos(\varphi
-\alpha)-{\rm i}\frac{\langle{\bf C}_3\rangle_\phi}{P\hbar}\ln\cot
\frac{\theta}{2}\Bigr)\sin^a\theta\tan^b\frac{\theta}{2}, \label{eq:3.3.4}
\end{equation}
where $A=A(\varphi)$ is an arbitrary periodic function of $\varphi$, and 
now the powers $a$ and $b$ are $a:=\alpha_3/\lambda\hbar-1$ and $b:=-\langle
{\bf p}(\alpha)\rangle_\phi/P\lambda\hbar$. 

Since the area element on ${\cal S}$ in the polar coordinates is ${\rm d}
{\cal S}=P^2\sin\theta\,d\theta\wedge d\varphi$ and the modulus of the first 
two factors in (\ref{eq:3.3.4}) is bounded on ${\cal S}$, the condition of 
the square-integrability of $\phi$ is equivalent to the existence of the 
integral $\int^\pi_0\sin^{2a+1}\theta\tan^{2b}(\theta/2){\rm d}\theta$. 
This condition is equivalent to $a>-1$ and $a+1>b>-a-1$; i.e. to 
\begin{equation}
\alpha_3>0, \hskip 25pt 
-P\alpha_3<\langle{\bf p}(\alpha)\rangle_\phi<P\alpha_3. 
\label{eq:3.3.5}
\end{equation}
Hence, by the requirement of the square-integrability of the wave function 
and $\alpha_3=\alpha_i\beta^i$, the angle between the directions $\alpha^i$ 
and $\beta^i$ must be zero or an acute angle; and the range of the 
expectation value $\langle{\bf p}(\alpha)\rangle_\phi$ is restricted by 
$\alpha_3$: $\vert\langle{\bf p}(\alpha)\rangle_\phi\vert<\alpha_3P\leq P$. 
However, in this interval $\langle{\bf p}(\alpha)\rangle_\phi$ is freely 
specifiable. 

Since $A=A(\varphi)$ is periodic, we may write it as $\sum_{m\in\mathbb{Z}}A_m
\exp({\rm i}m(\varphi-\alpha))$ with some complex constants $A_m$. Then, 
using (\ref{eq:2.6}), the normalization condition for $\phi$ is 
\begin{equation}
1=\int_{\cal S}\vert\phi\vert^2{\rm d}{\cal S}=2\pi P^2\sum_{m\in\mathbb{Z}}
\sum_{k=0}^\infty A_{k+m}\overline{A_m}\int^\pi_0I_k\sin^{2a+1}\theta\tan^{2b}
\frac{\theta}{2}\,{\rm d}\theta. \label{eq:3.3.6}
\end{equation}
Here $I_k(x)$ is the $k$th modified Bessel function of the first kind, and 
its argument now is $2\theta\sqrt{1-\alpha^2_3}/(\lambda\hbar)$. In 
particular, if in the expansion of $A(\varphi)$ there were only one non-zero 
coefficient, say $A_m$, then only $I_0$ would appear in (\ref{eq:3.3.6}). 

Then we can calculate the expectation values and uncertainties of the basic 
observables. However, for states with general $A(\varphi)$, these 
calculations are rather lengthy and technically involved, without yielding 
much more insight into the nature of the problem than in the special case 
when $A(\varphi)$ has the form $A_m\exp({\rm i}m(\varphi-\alpha))$ for a 
single $m$. Thus, for the sake of simplicity, we assume that $A(\varphi)$ 
has this special form. 

A direct consequence of (\ref{eq:3.3.2}) is that $\Delta_\phi{\bf C}(\beta)=
\Delta_\phi{\bf p}(\alpha)/\lambda$. Thus, we need to calculate only $\Delta
_\phi{\bf p}(\alpha)$. However, the calculation of $\langle({\bf p}(\alpha))
^2\rangle_\phi$ is a bit more complicated: using 
\begin{eqnarray*}
\hskip -20pt &{}&(\alpha_ip^i)^2=P^2\bigl(\sqrt{1-\alpha^2_3}\sin\theta\cos
 (\varphi-\alpha)+\alpha_3\cos\theta\bigr)^2 \\
\hskip -20pt &{}&=P^2\Bigl((1-\alpha^2_3)\frac{1}{2}\sin^2\theta\bigl(1+\cos
 (2(\varphi-\alpha)\bigr)+2\alpha_3\sqrt{1-\alpha^2_3}\cos\theta\sin\theta
 \cos(\varphi-\alpha)+\alpha^2_3\cos^2\theta\Bigr)
\end{eqnarray*}
and equation (\ref{eq:2.3}) for the modified Bessel functions, we find that 
\begin{eqnarray*}
\langle({\bf p}(\alpha))^2\rangle_\phi=2\pi P^4\vert A_m\vert^2\int^\pi_0\Bigl(
 \bigl(\!\!\!\!&\frac{1}{2}\!\!\!\!&(1-\alpha^2_3)\sin^2\theta+\alpha^2_3\cos^2
 \theta\bigr)\,I_0+2\alpha_3\sqrt{1-\alpha^2_3}\sin\theta\cos\theta\,I_1 \\
+\!\!\!\!&\frac{1}{2}\!\!\!\!&(1-\alpha^2_3)\sin^2\theta I_2\Bigr)\sin^{2a+1}
 \theta\tan^{2b}\frac{\theta}{2}\,{\rm d}\theta. 
\end{eqnarray*}
Then, expressing $I_1$ and $I_2$ by $I_0$ and its derivatives, by integration 
by parts, and using elementary trigonometric identities, we obtain 
\begin{eqnarray}
&{}&(\Delta_\phi{\bf p}(\alpha))^2=\langle({\bf p}(\alpha))^2\rangle_\phi-
  \bigl(\langle{\bf p}(\alpha)\rangle_\phi\bigr)^2 \label{eq:3.3.7} \\
&{}&=\frac{1}{2}\lambda\hbar\Bigl(\alpha_3P^2-2\pi P^4\vert A_m\vert^2\int
  ^\pi_0\bigl(\frac{1}{2}\lambda\hbar\sin^2\theta+\frac{\langle{\bf p}(\alpha)
  \rangle_\phi}{P}\cos\theta\bigr)\,I_0\sin^{2a+1} \theta\tan^{2b}\frac{\theta}
  {2}\,{\rm d}\theta\Bigr). \nonumber
\end{eqnarray}
$\langle{\bf p}(\alpha){\bf p}(\beta)\rangle_\phi=\langle{\bf p}(\alpha)
{\bf p}_3\rangle_\phi$ can be calculated in a similar way. It is given by 
\begin{equation}
\langle{\bf p}(\alpha){\bf p}(\beta)\rangle_\phi=2\pi P^4\vert A_m\vert^2\int
^\pi_0\bigl(\frac{1}{2}\lambda\hbar\sin^2\theta+\frac{\langle{\bf p}(\alpha)
\rangle_\phi}{P}\cos\theta\bigr)I_0\sin^{2a+1}\theta\tan^{2b}\frac{\theta}{2}
{\rm d}\theta, \label{eq:3.3.8}
\end{equation}
which is just the second term between the big round brackets in 
(\ref{eq:3.3.7}). Thus, by $\Delta_\phi{\bf C}(\beta)=\Delta_\phi{\bf p}
(\alpha)/\lambda$ and equations (\ref{eq:3.3.7})-(\ref{eq:3.3.8}), we have 
the explicit expression for the product uncertainty, too: 
\begin{equation*}
\Delta_\phi{\bf p}(\alpha)\Delta_\phi{\bf C}(\beta)=\frac{1}{2}\hbar\Bigl(
\alpha_i\beta^iP^2-\langle{\bf p}(\alpha){\bf p}(\beta)\rangle_\phi\Bigr). 
\end{equation*}
This depends on $\lambda$ only through the state $\phi$. 

To clarify the dependence of $\Delta_\phi{\bf p}(\alpha)\Delta_\phi{\bf C}
(\beta)$ on $\lambda$ with given $\langle{\bf p}(\alpha)\rangle_\phi$, let us 
recall that, by the normalization condition, $\vert A_m\vert^2$ depends on 
$\lambda$, too. Hence, by (\ref{eq:3.3.6}) and (\ref{eq:3.3.8}) 
\begin{equation}
\langle{\bf p}(\alpha){\bf p}(\beta)\rangle_\phi=P^2\frac{\int_0^\pi I_0\bigl(
\frac{\langle{\bf p}(\alpha)\rangle_\phi}{P}\cos\theta+\frac{1}{2}\lambda\hbar
\sin^2\theta\bigr)\sin^{2a+1}\theta\tan^{2b}\frac{\theta}{2}{\rm d}\theta}
{\int_0^\pi I_0\sin^{2a+1}\theta\tan^{2b}\frac{\theta}{2}{\rm d}\theta}.
\label{eq:3.3.9}
\end{equation}
Since $I_0\sin^{2a+1}\theta\tan^{2b}(\theta/2)$ is non-negative on the 
interval $[0,\pi]$ and $\sin^2\theta$, $\cos\theta\leq1$, the integrand in 
the numerator of (\ref{eq:3.3.9}) is not greater than $(\vert\langle
{\bf p}(\alpha)\rangle_\phi\vert/P+\lambda\hbar/2)I_0\sin^{2a+1}\theta$ $\tan
^{2b}(\theta/2)$; and it is not less than $-(\vert\langle{\bf p}(\alpha)
\rangle_\phi\vert/P)I_0\sin^{2a+1}\theta$ $\tan^{2b}(\theta/2)$. Hence, 
$-P\vert\langle{\bf p}(\alpha)\rangle_\phi\vert\leq\langle{\bf p}(\alpha)
{\bf p}(\beta)\rangle_\phi$ $\leq P\vert\langle{\bf p}(\alpha)\rangle_\phi
\vert+\lambda\hbar P^2/2$, yielding 
\begin{equation*}
\frac{1}{2}\hbar P\Bigl(\alpha_3P-\vert\langle{\bf p}(\alpha)\rangle_\phi
\vert-\frac{1}{2}\lambda\hbar P\Bigr)\leq\Delta_\phi{\bf p}(\alpha)\Delta_\phi
{\bf C}(\beta)\leq\frac{1}{2}\hbar P\Bigl(\alpha_3P+\vert\langle{\bf p}
(\alpha)\rangle_\phi\vert\Bigr).
\end{equation*}
Since by (\ref{eq:3.3.5}) $\alpha_3P-\vert\langle{\bf p}(\alpha)\rangle
_\phi\vert$ is strictly positive, by $\Delta_\phi{\bf p}(\alpha)=\lambda
\Delta_\phi{\bf C}(\beta)$ this implies that, in the $\lambda\to0$ limit, 
$\Delta_\phi{\bf C}(\beta)$ diverges and $\Delta_\phi{\bf p}(\alpha)$ tends to 
zero; and $\lim_{\lambda\to0}\Delta_\phi{\bf p}(\alpha)\Delta_\phi{\bf C}(\beta)$ 
is not less than $\frac{1}{2}\hbar P(\alpha_3P-\vert\langle{\bf p}(\alpha)
\rangle_\phi\vert)>0$ but less than $\hbar P^2\alpha_3$. Hence this limit is 
finite and positive. 

The calculation of the $\lambda\to\infty$ limit is a bit longer. By 
(\ref{eq:2.4}) 
\begin{equation*}
I_0\sin^{2(a+1)}\theta\tan^{2b}\frac{\theta}{2}=\Bigl(1+\bigl(\frac{\theta
\sqrt{1-\alpha^2_3}}{\lambda\hbar}\bigr)^2+O(\lambda^{-4})\Bigr)2^{2(a+1)}
\bigl(\sin^2\frac{\theta}{2}\bigr)^{a+b+1}\bigl(\cos^2\frac{\theta}{2}\bigr)
^{a-b+1},
\end{equation*}
where the powers are 
\begin{equation*}
a+b+1=\frac{1}{\lambda\hbar P}\bigl(\alpha_3P-\langle{\bf p}(\alpha)\rangle
_\phi\bigr)=:\frac{A}{\lambda}>0, \hskip 20pt
a-b+1=\frac{1}{\lambda\hbar P}\bigl(\alpha_3P+\langle{\bf p}(\alpha)\rangle
_\phi\bigr)=:\frac{B}{\lambda}>0.
\end{equation*}
Hence, for large $\lambda$, the denominator of (\ref{eq:3.3.9}) is 
\begin{eqnarray*}
&{}&\int^\pi_0I_0\sin^{2a+1}\theta\tan^{2b}\frac{\theta}{2}{\rm d}\theta=\int
  ^\pi_0\frac{1}{\sin\theta}\sin^{2(a+1)}\theta\tan^{2b}\frac{\theta}{2}{\rm d}
  \theta+O(\lambda^{-2}) \\
&{}&=\frac{1}{2}2^{2(a+1)}\int^{\pi/2}_0(\sin x)^{2a+2b+1}(\cos x)^{2a-2b+1}{\rm d}x
  +O(\lambda^{-2}) \\
&{}&=\frac{1}{2}2^{2(a+1)}\beta\bigl(\frac{A}{\lambda},\frac{B}{\lambda}\bigr)
  +O(\lambda^{-2})=\frac{1}{2}2^{2(a+1)}\frac{\Gamma(\frac{A}{\lambda})\Gamma
  (\frac{B}{\lambda})}{\Gamma(\frac{A+B}{\lambda})}+O(\lambda^{-2}),
\end{eqnarray*}
where $\beta(x,y)$ is Euler's beta function, which can also be re-expressed 
by the $\Gamma$ function as above (see \cite{AbSt}, p. 258). (Here, to avoid 
confusion, we denoted the beta function by $\beta(x,y)$ instead of the 
standard $B(x,y)$.) Using also the property $\Gamma(x+1)=x\Gamma(x)$ of 
the $\Gamma$ function, in a similar way we obtain that the integrals in the 
numerator of (\ref{eq:3.3.9}) are 
\begin{eqnarray*}
&{}&\int^\pi_0I_0\sin\theta\sin^{2(a+1)}\theta\tan^{2b}\frac{\theta}{2}
  {\rm d}\theta=2\frac{AB}{(A+B)(A+B+\lambda)}2^{2(a+1)}\frac{\Gamma(
  \frac{A}{\lambda})\Gamma(\frac{B}{\lambda})}{\Gamma(\frac{A+B}{\lambda})}
  +O(\lambda^{-2}), \\
&{}&\int^\pi_0I_0\frac{\cos\theta}{\sin\theta}\sin^{2(a+1)}\theta\tan^{2b}
  \frac{\theta}{2}{\rm d}\theta=\frac{1}{2}\frac{B-A}{B+A}2^{2(a+1)}
  \frac{\Gamma(\frac{A}{\lambda})\Gamma(\frac{B}{\lambda})}{\Gamma(
  \frac{A+B}{\lambda})}+O(\lambda^{-2}).
\end{eqnarray*}
Substituting these into (\ref{eq:3.3.9}) we find that 
\begin{equation*}
\langle{\bf p}(\alpha){\bf p}(\beta)\rangle_\phi=\alpha_3P^2-\frac{2}{\hbar}
\Bigl(\alpha^2_3P^2-\langle{\bf p}(\alpha)\rangle_\phi^2\Bigr)\frac{1}
{\lambda}+O(\lambda^{-2}).
\end{equation*}
Hence, in the $\lambda\to\infty$ limit, $(\Delta_\phi{\bf p}(\alpha))^2$ 
tends to $\alpha^2_3P^2-\langle{\bf p}(\alpha)\rangle_\phi^2>0$, while both 
$\Delta_\phi{\bf C}(\beta)$ and $\Delta_\phi{\bf p}(\alpha)\Delta_\phi{\bf C}
(\beta)$ tend to zero as $1/\lambda$. 

Therefore, to summarize, by (\ref{eq:3.3.5}) and $\alpha_3=\alpha_i\beta^i$, 
the most classical states exist for ${\bf p}(\alpha)$ and ${\bf C}(\beta)$ 
\emph{only if} the angle between the directions $\alpha^i$ and $\beta^i$ is 
\emph{zero or an acute angle}. If this is the right or a blunt angle, then 
\emph{no} such state exists. There is no restriction on the spin $s$. The 
solutions depend on one free function, and the solution even with $A=A_m\exp
({\rm i}m(\varphi-\alpha))$ forms a three-parameter family, parameterized by 
the expectation values, $\langle{\bf p}(\alpha)\rangle_\phi$ and $\langle
{\bf C}(\beta)\rangle_\phi$, and by $\lambda$. However, while the range of 
the expectation value $\langle{\bf C}(\beta)\rangle_\phi$ and the positive 
parameter $\lambda$ is unrestricted, $\langle{\bf p}(\alpha)\rangle_\phi$ is 
restricted by (\ref{eq:3.3.5}), but no more. The closer the $\alpha^i$ to be 
orthogonal to $\beta^i$, the shorter the interval around zero from which 
$\langle{\bf p}(\alpha)\rangle_\phi$ can take values. In the limit $\alpha_3
\to0$, i.e. when $\alpha^i$ is getting to be orthogonal to $\beta^i$, the 
second condition in (\ref{eq:3.3.5}) reduces to $\langle{\bf p}(\alpha)
\rangle_\phi\to0$ (compare with the conditions (\ref{eq:3.2.5}) in the 
$({\bf p}^i,{\bf J}_i)$ system). As functions of $\lambda$, the asymptotic 
behaviour of the standard deviations and the product uncertainty is similar 
to that of the standard deviations and the product uncertainty in the 
$({\bf p}^i,{\bf J}_i)$ case. 


\subsubsection{The non-existence of the most classical states for the pairs 
of observables $({\bf C}(\alpha),{\bf C}(\beta))$}
\label{sub-3.3.2}

The philosophy of the calculation is similar to the previous ones: we form 
${\bf C}(\alpha):=\alpha^i{\bf C}_i$ and ${\bf C}(\beta):=\beta^i{\bf C}_i$, 
in which $\beta^i=(0,0,1)$ and the unit vector $\alpha_i$ is not parallel 
with $\beta_i$, i.e. $\alpha_3\not=\pm1$. For these observables 
\begin{equation}
\Delta_\phi{\bf C}(\alpha)\,\Delta_\phi{\bf C}(\beta)\geq\frac{\hbar}{2}
\vert\hbar Ps\,\alpha^i\beta^j\varepsilon_{ijk}\langle{\bf p}^k\rangle_\phi-
\alpha^i\beta^j\langle{\bf C}_i{\bf p}_j-{\bf C}_j{\bf p}_i\rangle_\phi\vert 
\label{eq:3.3.11}
\end{equation}
follows. Here, the equality holds precisely when 
\begin{equation}
\Bigl({\bf C}(\alpha)-{\rm i}\lambda{\bf C}(\beta)\Bigr)\phi=\Bigl(\langle
{\bf C}(\alpha)\rangle_\phi-{\rm i}\lambda\langle{\bf C}(\beta)\rangle_\phi
\Bigr)\phi \label{eq:3.3.12}
\end{equation}
for some $\lambda>0$. This, in the unitary, irreducible representation 
labeled by $P$ and $s$, is equivalent to 
\begin{equation}
P\bigl(\alpha^i-{\rm i}\lambda\beta^i\bigr)\Bigl(m_i{\edth}'\phi+\bar m_i
{\edth}\phi-\frac{p_i}{P^2}\phi\Bigr)=-\frac{\rm i}{P\hbar}\bigl(\langle
{\bf C}(\alpha)\rangle_\phi-{\rm i}\lambda\langle{\bf C}(\beta)\rangle_\phi
\bigr)\phi=:-{\rm i}C\phi. \label{eq:3.3.13}
\end{equation}
Since ${\bf C}(\alpha)-{\rm i}\lambda{\bf C}(\beta)$ does not commute with 
the Casimir operator ${\bf J}_i{\bf J}^i$ of the subalgebra $su(2)\subset 
e(3)$, the spinorial method of \cite{Sz1} cannot be used directly to solve 
the eigenvalue problem (\ref{eq:3.3.12}). Thus, our strategy might be to 
find first the general \emph{local} solution $\phi$ of (\ref{eq:3.3.13}), 
and then to determine the eigenvalue $C$ from the requirement that $\phi$ be 
globally well defined and square integrable on ${\cal S}$. In fact, this 
strategy could have been followed successfully in the determination of the 
most classical states for the observables $({\bf J}(\alpha),{\bf J}
(\beta))$ in subsection \ref{sub-3.2.2} (or in \cite{Sz1}). Nevertheless, 
following this strategy, we show that \emph{there is no such states 
represented by any square integrable wave function that would also be 
differentiable on ${\cal S}$ except only at isolated points}. 

\bigskip
\noindent
$\bullet$ {\bf The general local solutions}

\noindent
As in subsection \ref{sub-3.2.2}, we write equation (\ref{eq:3.3.13}) in the 
rotated complex stereographic coordinates $(\xi,\bar\xi)$. It is 
\begin{equation}
Y(\ln\phi)=-{\rm i}C-\frac{1}{2}s\sqrt{1-\alpha_3^2}(\xi-\bar\xi)+\frac{1}{1+
\xi\bar\xi}\Bigl(\sqrt{1-\alpha_3^2}(\xi+\bar\xi)+(\alpha_3-{\rm i}\lambda)
\bigl(\xi\bar\xi-1\bigr)\Bigr), \label{eq:3.3.14}
\end{equation}
where the complex vector field $Y$ on its left hand side is defined by 
\begin{equation}
Y:=\Bigl(\frac{1}{2}\sqrt{1-\alpha_3^2}(1-\xi^2)+(\alpha_3-{\rm i}\lambda)
\xi\Bigr)\frac{\partial}{\partial\xi}+\Bigl(\frac{1}{2}\sqrt{1-\alpha_3^2}
(1-\bar\xi^2)+(\alpha_3-{\rm i}\lambda)\bar\xi\Bigr)\frac{\partial}
{\partial\bar\xi}. \label{eq:3.3.15}
\end{equation}
It might be worth noting that the vector field $Y$ is orthogonal to $X$, 
introduced in subsection \ref{sub-3.2.2}, with respect to the metric 
(\ref{eq:A.1.4}) on ${\cal S}$. Also, it has vanishing Lie bracket with $X$, 
and hence, in particular, the functions $v$ and $u$ defined below and 
satisfying $X(u)=1$, $X(v)=0$, $Y(u)=0$ and $Y(v)=1$ form a local complex 
coordinate system adapted to the vector fields $Y$ and $X$ on the 
complexified ${\cal S}$. 

We solve (\ref{eq:3.3.14}) by rewriting the terms on the right in the form 
of derivatives in the direction $Y$. Since 
\begin{equation*}
-\bigl(\sqrt{1-\alpha^2_3}\xi-\alpha_3+{\rm i}\lambda\bigr)=Y\Bigl(\ln\bigl(
\frac{1}{2}\sqrt{1-\alpha^2_3}(1-\xi^2)+(\alpha_3-{\rm i}\lambda)\xi\bigr)
\Bigr),
\end{equation*}
the second term on the right of (\ref{eq:3.3.14}) can be written in the form 
\begin{equation*}
Y\Bigl(\ln\bigl(\frac{(\xi-\xi_+)(\xi-\xi_-)}{(\bar\xi-\xi_+)(\bar\xi-\xi_-)}
\bigr)^{s/2}\Bigr),
\end{equation*}
where we used the notation introduced in (\ref{eq:3.2.14}). Also, a simple 
calculation shows that 
\begin{eqnarray*}
Y\bigl(\ln(1+\xi\bar\xi)\bigr)\!\!\!\!&=\!\!\!\!&\sqrt{1-\alpha^2_3}\frac{\xi
 +\bar\xi}{1+\xi\bar\xi}+(\alpha_3-{\rm i}\lambda)\frac{\xi\bar\xi-1}
 {1+\xi\bar\xi}+ \\
\!\!\!\!&+\!\!\!\!&\frac{1}{2}Y\Bigl(\ln\bigl((\xi-\xi_+)(\xi-\xi_-)(\bar\xi
 -\xi_+)(\bar\xi-\xi_-)\bigr)\Bigr),
\end{eqnarray*}
and hence the last term on the right of (\ref{eq:3.3.14}) is 
\begin{equation*}
Y\Bigl(\ln\frac{1+\xi\bar\xi}{\sqrt{(\xi-\xi_+)(\xi-\xi_-)(\bar\xi-\xi_+)
(\bar\xi-\xi_-)}}\Bigr). 
\end{equation*}
Thus, if we knew the solution $v$ of $Y(v)=1$, then the first term on the 
right of (\ref{eq:3.3.14}) would have the form $Y(-{\rm i}Cv)$ and hence
we already would have a \emph{particular} solution of (\ref{eq:3.3.14}). 
To have its \emph{general} solution, we need the general solution $\phi_0$ 
of the \emph{homogeneous} equation $Y(\phi_0)=0$, too. 

In the generic case (i.e. when $(1-\lambda)^2+\alpha^2_3>0$, see subsection 
\ref{sub-3.2.2}), the solution of $Y(v)=1$ is precisely the function $v:=
u_1+u_2$ (up to the addition of the solution of the homogeneous equation 
$Y(\phi_0)=0$), where $u_1$ and $u_2$ are given by 
\begin{equation}
u_1=-\frac{1}{2\sqrt{1-\lambda^2-2{\rm i}\lambda\alpha_3}}\ln\frac{\xi-\xi_+}
{\xi-\xi_-}, \hskip 20pt
u_2=-\frac{1}{2\sqrt{1-\lambda^2-2{\rm i}\lambda\alpha_3}}\ln\frac{\bar\xi-
\xi_+}{\bar\xi-\xi_-}. \label{eq:3.3.16}
\end{equation}
The \emph{general} solution $\phi_0$ of $Y(\phi_0)=0$ is an arbitrary smooth 
complex function of $u:=u_1-u_2$, i.e. $\phi_0$ is an arbitrary smooth 
complex function of 
\begin{equation}
z:=\frac{(\xi-\xi_-)(\bar\xi-\xi_+)}{(\bar\xi-\xi_-)(\xi-\xi_+)}.
\label{eq:3.3.17}
\end{equation}
Thus, $\phi_0$ depends on $\xi$ and $\bar\xi$ \emph{only through} $z$: $\phi
_0=\phi_0(z(\xi,\bar\xi))$. Therefore, in the generic case, the general 
\emph{local} solution of equation (\ref{eq:3.3.14}) is 
\begin{eqnarray}
\phi\!\!\!\!&=\!\!\!\!&\phi_0\Bigl(\frac{(\xi-\xi_+)(\bar\xi-\xi_+)}
  {(\xi-\xi_-)(\bar\xi-\xi_-)}\Bigr)^{M/2}\Bigl(\frac{(\xi-\xi_+)(\xi-\xi_-)}
  {(\bar\xi-\xi_+) (\bar\xi-\xi_-)}\Bigr)^{s/2}\frac{1+\xi\bar\xi}
  {\sqrt{(\xi-\xi_+)(\xi-\xi_-)(\bar\xi-\xi_+)(\bar\xi-\xi_-)}} \nonumber \\
\!\!\!\!&=\!\!\!\!&\phi_0\bigl(1+\xi\bar\xi\bigr)\bigl(\xi-\xi_+)^{\frac{1}{2}
  (M+s-1)}(\bar\xi-\xi_+)^{\frac{1}{2}(M-s-1)}(\xi-\xi_-)^{\frac{1}{2}(-M+s-1)}
  (\bar\xi-\xi_-)^{\frac{1}{2}(-M-s-1)}, \hskip 25pt\label{eq:3.3.18}
\end{eqnarray}
where the power $M$ is defined by 
\begin{equation}
M:=\frac{{\rm i}C}{\sqrt{1-\lambda^2-2{\rm i}\lambda\alpha_3}}. 
\label{eq:3.3.19}
\end{equation}
Note that, in general, this might be complex: $M=M_1+{\rm i}M_2$, where $M_1,
M_2\in\mathbb{R}$. 

In the exceptional case (i.e. when $\alpha_3=0$ and $\lambda=1$), the 
solution of $Y(v)=1$ is $v=u_1+u_2$, while that of $Y(\phi_0)=0$ is an 
arbitrary smooth complex function of $u=u_1-u_2$, where now 
\begin{equation}
u_1:=\frac{1}{{\rm i}+\xi}, \hskip 25pt
u_2:=\frac{1}{{\rm i}+\bar\xi}. \label{eq:3.3.20}
\end{equation}
Hence, the \emph{local} solution of (\ref{eq:3.3.14}) is 
\begin{equation}
\phi=\phi_0\exp\Bigl(-{\rm i}C\bigl(\frac{1}{{\rm i}+\xi}+\frac{1}{{\rm i}+
\bar\xi}\bigr)\Bigr)\Bigl(\frac{{\rm i}+\xi}{{\rm i}+\bar\xi}\Bigr)^s\frac{1
+\xi\bar\xi}{({\rm i}+\xi)({\rm i}+\bar\xi)}, \label{eq:3.3.21}
\end{equation}
where $\phi_0$ is still an arbitrary smooth complex function of $u$, and 
hence $\phi_0$ depends on $\xi$ and $\bar\xi$ only through $u$.

\bigskip
\noindent
$\bullet$ {\bf The non-existence of differentiable wave functions: The generic
case}

\noindent
The coefficient of $\phi_0$ in (\ref{eq:3.3.18}), i.e. 
\begin{equation}
F:=\bigl(1+\xi\bar\xi\bigr)\bigl(\xi-\xi_+)^{\frac{1}{2}(M+s-1)}(\bar\xi-\xi_+)
^{\frac{1}{2}(M-s-1)}(\xi-\xi_-)^{\frac{1}{2}(-M+s-1)}(\bar\xi-\xi_-)^{\frac{1}{2}(-M-s-1)},
\label{eq:3.3.22}
\end{equation}
is bounded in the $\xi\to\infty$ limit, but, depending on the value of $s$ 
and $M$, it is either singular or finite (or zero) at the points $\xi_+$, 
$\bar\xi_+$, $\xi_-$ and $\bar\xi_-$. In addition, $F$ may have 
discontinuities along lines in the $(\xi,\bar\xi)$-plane. Then the role of 
the still free function $\phi_0$ of $z$ in (\ref{eq:3.3.18}) would be to 
ensure the square integrability of $\phi$ in such a way that the zeros of 
$\phi_0$ would compensate the singularities of $F$ to obtain square 
integrable $\phi$. ($\phi_0$ may have singularities at the zeros of $F$, 
and may have jumps to compensate the potential discontinuities of $F$.) We 
show that this strategy yields that \emph{there is no $\phi_0$ that could 
yield a square integrable $\phi$ which would, in addition, be differentiable 
everywhere on ${\cal S}$ except possibly at the points $\xi_\pm$, $\bar\xi
_\pm$}. 

First we show that $(\xi,\bar\xi)\mapsto z$, given by (\ref{eq:3.3.17}), is 
a $\mathbb{C}\to\mathbb{C}$ \emph{surjective} map, and hence the domain of 
$\phi_0$ should be the entire complex $z$-plane except perhaps curves or 
finitely many isolated points. By (\ref{eq:3.3.17}) the value $z=1$ is the 
image of the whole real axis $\bar\xi=\xi$. Hence we should show only that 
any $z\not=1$ is the image of some $\xi$. Using $\xi_+\xi_-=-1$ and the 
notation $\xi=:r\exp({\rm i}\chi)$, (\ref{eq:3.3.17}) gives 
\begin{eqnarray*}
&{}&r^2\bigl(z-1\bigr)-r\exp({\rm i}\chi)\bigl(z\xi_--\xi_+\bigr)-
  r\exp(-{\rm i}\chi)\bigl(z\xi_+-\xi_-\bigr)-\bigl(z-1\bigr)=0, \\
&{}&r^2\bigl(\bar z-1\bigr)-r\exp(-{\rm i}\chi)\bigl(\bar z\bar\xi_--\bar
  \xi_+\bigr)-r\exp({\rm i}\chi)\bigl(\bar z\bar\xi_+-\bar\xi_-\bigr)-
  \bigl(\bar z-1\bigr)=0.
\end{eqnarray*}
Using $z\not=1$, these yield that 
\begin{eqnarray*}
&{}&\exp(2{\rm i}\chi)=\frac{(z-1)(\bar z\bar\xi_--\bar\xi_+)-(\bar z-1)
 (z\xi_+-\xi_-)}{(\bar z-1)(z\xi_--\xi_+)-(z-1)(\bar z\bar\xi_+-\bar\xi_-)},\\
&{}&r^2-\frac{1}{2}r\Bigl(\exp({\rm i}\chi)\bigl(\frac{z\xi_--\xi_+}{z-1}-
 \frac{\bar z\bar\xi_+-\bar\xi_-}{\bar z-1}\bigr)+\exp(-{\rm i}\chi)\bigl(
 \frac{z\xi_+-\xi_-}{z-1}-\frac{\bar z\bar\xi_--\bar\xi_+}{\bar z-1}\bigr)
 \Bigr)-1=0.
\end{eqnarray*}
The first is an explicit expression of the phase of $\xi$ in terms of $z$ 
and the constants $\xi_\pm$ without any further restriction on them; and the 
second has a unique (positive) solution since its discriminant is positive. 
Hence, any $z\in\mathbb{C}$ has some pre-image under the map (\ref{eq:3.3.17}). 
Therefore, $\phi_0$ should be defined on the whole complex plane (except
possibly along lines and at isolated points determined by $F$) and depends 
only on $z$, but \emph{not} on $\bar z$. Note also that $z$ maps the points 
$\xi_-,\bar\xi_+$ of the complex $\xi$-plane into the zero of the complex 
$z$-plane, it maps the whole real axis $\bar\xi=\xi$ into $1$, and the points 
$\xi_+,\bar\xi_-$ to infinity. Also, the point $z=1$ corresponds to the 
infinity $\xi=\infty$ of the $\xi$-plane. (N.B.: The map $(\xi,\bar\xi)
\mapsto z$ is the product of one fraction linear transformation and the 
complex conjugate of another one.) 

By (\ref{eq:3.2.14}) it is easy to see that, in the generic case, $\xi_+$, 
$\bar\xi_+$, $\xi_-$ and $\bar\xi_-$ are different points of the complex 
$\xi$-plane and none of them is zero. Hence, if $\xi_0$ denotes any of these 
points, then there is a positive number, $R>0$, such that $U_R(\xi_0):=
\{\xi\in\mathbb{C}\vert\,\vert\xi-\xi_0\vert<R\}$ is an open neighbourhood 
of $\xi_0$, which does not contain any of the others. Since ${\cal S}$ is 
compact, the square integrability of $\phi$ is equivalent to its local 
square integrability. Since $F$ is finite on the real axis where $z=1$, 
moreover $z\to1$ if $\xi\to\infty$, in the $\xi\to\infty$ limit $\phi_0(z
(\xi,\bar\xi))\to\phi_0(1)$. Thus, $\phi_0$ remains bounded in the $\xi\to
\infty$ limit, too. Therefore, by (\ref{eq:3.3.22}) the square integrability 
of $\phi$ is equivalent to the square integrability of $\phi_0(\xi-\xi_\pm)
^p$ and of $\phi_0(\bar\xi-\xi_\pm)^p$ on $U_R(\xi_\pm)$, where $p$ is the 
corresponding (in general complex) power in $F$. Nevertheless, although 
$(\xi,\bar\xi)\mapsto\phi_0(z(\xi,\bar\xi))$ is smooth even on the real axis 
of the complex $\xi$-plane, $z\mapsto\phi_0(z)$ is \emph{not} necessarily 
differentiable at $z=1$, because the map $(\xi,\bar\xi)\mapsto z(\xi,\bar
\xi)$ shrinks the whole real axis and the point at infinity of the complex 
$\xi$-plane to the single point $1$ of the $z$-plane. 

Since $\xi_+$, $\bar\xi_+$, $\xi_-$ and $\bar\xi_-$ are all different points 
and none of them is zero, by (\ref{eq:3.3.17}) $z$ depends in the limit $\xi
\to\xi_\pm$ essentially only on $\xi-\xi_\pm$, and hence in this limit $\phi
_0=\phi_0(z(\xi,\bar\xi))$ depends essentially only on $\xi-\xi_\pm$. In a 
similar way, in the limit $\bar\xi\to\xi_\pm$, the function $\phi_0(z(\xi,
\bar\xi))$ depends essentially only on $\bar\xi-\xi_\pm$. Thus first let us 
suppose that $\phi_0\simeq\alpha(\xi-\xi_\pm)^q$ in the limit $\xi\to\xi_\pm$ 
for some $q,\alpha\in\mathbb{C}$. If we write $\xi-\xi_\pm=r\exp({\rm i}
\chi)$ and $p=:p_1+{\rm i}p_2$ and $q=:q_1+{\rm i}q_2$ with $p_1,p_2,q_1,
q_2\in\mathbb{R}$, then 
\begin{equation}
\phi_0(\xi-\xi_\pm)^p\simeq\alpha(\xi-\xi_\pm)^{p+q}=\alpha r^{p_1+q_1}\exp\bigl(
-(p_2+q_2)\chi\bigr)\exp\Bigl({\rm i}\bigl((p_1+q_1)\chi+(p_2+q_2)\ln r\bigr)
\Bigr). \label{eq:3.3.23}
\end{equation}
Hence, on the neighbourhood $U_R(\xi_\pm)$, one has $\vert\phi_0(\xi-\xi_\pm)^p
\vert^2\simeq\vert\alpha\vert^2r^{2(p_1+q_1)}\exp(-2(p_2+q_2)\chi)$. Then the 
$\epsilon\to0$ limit of 
\begin{equation*}
\int^R_\epsilon\int^{2\pi}_0\vert\phi_0(\xi-\xi_\pm)^p\vert^2 r\,d\chi\,dr=
\frac{1-\exp\bigl(-4\pi(p_2+q_2)\bigl)}{2(p_2+q_2)}\frac{\vert\alpha\vert^2}
{2(p_1+q_1+1)}\bigl[r^{2(p_1+q_1+1)}\bigr]^R_\epsilon
\end{equation*}
is finite precisely when $q_1+p_1+1>0$. Since for sufficiently small $R$ the 
deviation of the metric area element (\ref{eq:A.1.4}) of ${\cal S}$ from 
the `coordinate area element' $r\,d\chi dr$ on $U_R(\xi_\pm)$ can be neglected, 
we obtain that $\phi_0(\xi-\xi_\pm)^p$ is square integrable on $U_R(\xi_\pm)$ 
precisely when $q_1>-p_1-1$. In a similar way, in a neighbourhood of $\bar
\xi_\pm$, i.e. in the $\bar\xi\to\xi_\pm$ limit, if $\phi_0\simeq\alpha
(\bar\xi-\xi_\pm)^q$, then for the square integrability of $\phi_0(\bar\xi-
\xi_\pm)^p$ on $U_R(\xi_\pm)$ we obtain the same condition: $q_1>-p_1-1$.

In particular, we have the following 
\begin{description}
\item[A.]
in the $\xi\to\xi_+$ limit if $\phi_0\simeq\alpha(\xi-\xi_+)^a$ with $a=:a_1+
{\rm i}a_2$, $a_1,a_2\in\mathbb{R}$, then the condition $q_1>-p_1-1=-
(M_1+s+1)/2$ above shows that there are two possibilities: i. if $M_1+s\leq
-1$, then, in the $\xi\to\xi_+$ limit, $\phi_0$ \emph{must} tend to zero as
$\alpha(\xi-\xi_+)^{a_1}$, where $a_1>-(M_1+s+1)/2\geq0$; or ii. if 
$M_1+s>-1$, then in this limit $\phi_0$ might in principle diverge, but there 
is an \emph{upper bound} on the rate of its divergence: $\phi_0\simeq\alpha
(\xi-\xi_+)^{a_1}$, where now $-a_1<(M_1+s+1)/2$. 

\item[B.]
in the $\xi\to\bar\xi_+$ limit if $\phi_0\simeq\alpha(\xi-\bar\xi_+)^b$ with 
$b=:b_1+{\rm i}b_2$, $b_1,b_2\in\mathbb{R}$, then either i. $M_1-s\leq-1$, 
and then in this limit $\phi_0$ \emph{must} tend to zero as $\alpha(\xi-\bar
\xi_+)^{b_1}$, where $b_1>-(M_1-s+1)/2\geq0$; or ii. $M_1-s>-1$, when 
$\phi_0$ might in principle diverge as $\alpha(\xi-\bar\xi_+)^{b_1}$, but 
there is an upper bound on the rate of its divergence, which is $-b_1<
(M_1-s+1)/2$. 

\item[C.]
in the $\xi\to\xi_-$ limit if $\phi_0\simeq\alpha(\xi-\xi_-)^c$ with $c=:c_1
+{\rm i}c_2$, $c_1,c_2\in\mathbb{R}$, then either i. $M_1-s\geq1$, in which 
case $\phi_0$ \emph{must} tend to zero as $\alpha(\xi-\xi_-)^{c_1}$, where 
$c_1>(M_1-s-1)/2\geq0$; or ii. $M_1-s<1$, in which case $\phi_0$ might 
diverge, but the upper bound on the rate of its divergence is $-c_1<
(-M_1+s+1)/2$. 

\item[D.]
in the $\xi\to\bar\xi_-$ limit if $\phi_0\simeq\alpha(\xi-\bar\xi_-)^d$ with 
$d=:d_1+{\rm i}d_2$, $d_1,d_2\in\mathbb{R}$, then either i. $M_1+s\geq1$, in 
which case $\phi_0$ \emph{must} tend to zero as $\alpha(\xi-\bar\xi_-)^{d_1}$ 
with $d_1>(M_1+s-1)/2\geq0$; or ii. $M_1+s<1$, when $\phi_0$ might 
diverge as $\alpha(\xi-\bar\xi_-)^{d_1}$, but there is an upper bound on the 
rate of its divergence: $-d_1<(-M_1-s+1)/2$. 
\end{description}
However, since $\phi_0$ depends on $\xi$ and $\bar\xi$ \emph{only through} 
$z$, and moreover by (\ref{eq:3.3.17}) $z\to0$ if $\xi\to\bar\xi_+$ or $\xi\to
\xi_-$ and $z\to\infty$ if $\xi\to\xi_+$ or $\xi\to\bar\xi_-$, the asymptotic 
behaviour of $\phi_0$ in the limits $\xi\to\bar\xi_+,\xi_-$ is the same, 
and in the limits $\xi\to\xi_+,\bar\xi_-$ also. Hence $a=d$ and $b=c$ 
must hold, and, in particular, 
\begin{equation}
\lim_{\xi\to\xi_+}\phi_0\bigl(z(\xi,\bar\xi)\bigr)=\lim_{\xi\to\bar\xi_-}\phi_0
  \bigl(z(\xi,\bar\xi)\bigr)  , \hskip 20pt
\lim_{\xi\to\bar\xi_+}\phi_0\bigl(z(\xi,\bar\xi)\bigr)=\lim_{\xi\to\xi_-}\phi_0
  \bigl(z(\xi,\bar\xi)\bigr). \label{eq:3.3.24}
\end{equation}
These constraints must be taken into account when we determine the 
asymptotic form of $\phi_0$ from the above results in cases 
${\bf A}$.--${\bf D}$. 

Next, let us evaluate the consequences of the requirement that the wave 
function $\phi$ be well defined on the whole ${\cal S}$ except only at the 
points $\xi_\pm,\bar\xi_\pm$. This means that $\phi$ must be periodic on 
small closed complex paths even if they surround (but do not cross) a 
singularity. On the path $\xi=\xi_\pm+r\exp({\rm i}\chi)$, $\chi\in[0,2\pi)$, 
for sufficiently small $r$ the wave function $\phi$ is periodic in $\chi$ 
with period $2\pi$ precisely when $\phi_0(\xi-\xi_\pm)^p$ is periodic. By
(\ref{eq:3.3.23}) this is equivalent to $q_2=-p_2$ and $q_1+p_1=n\in
\mathbb{Z}$. However, by the condition $q_1>-p_1-1$ of the local square 
integrability, the second condition is, in fact, $p_1+q_1=n=0,1,2,...$. The 
analogous argumentation in the case of paths surrounding the singularities 
$\bar\xi_\pm$ gives the same result. Therefore, in particular cases 
${\bf A}$.--${\bf D}$, we have that 
\begin{eqnarray*}
&{}&a_1=-\frac{1}{2}\bigl(M_1+s+1\bigr)+n_1, \hskip 20pt 
  a_2=-\frac{1}{2}M_2, \\
&{}&b_1=-\frac{1}{2}\bigl(M_1-s+1\bigr)+n_2, \hskip 20pt 
  b_2=-\frac{1}{2}M_2, \\
&{}&c_1=\frac{1}{2}\bigl(M_1-s+1\bigr)+n_3, \hskip 27pt 
  c_2=\frac{1}{2}M_2, \\
&{}&d_1=\frac{1}{2}\bigl(M_1+s+1\bigr)+n_4, \hskip 27pt 
  d_2=\frac{1}{2}M_2,
\end{eqnarray*}
where $n_1,n_2,n_3,n_4=0,1,2,...$. But by $a=d$ and $b=c$ these immediately 
imply that both $M$ and the powers $a,b,c,d$ must be real; that 
\begin{equation}
a=\frac{1}{2}\bigl(1+n_1+n_4\bigr)\geq\frac{1}{2}, \hskip 20pt
b=\frac{1}{2}\bigl(1+n_2+n_3\bigr)\geq\frac{1}{2}; \label{eq:3.3.25}
\end{equation}
and that $2M=n_1+n_2-n_3-n_4$ and $2s=n_1-n_2+n_3-n_4$ hold. Thus, both 
$M\pm s$ are integer. Therefore, $\phi_0(z)$ must tend to zero as $1/z^a$ 
if $z\to\infty$, and also to zero as $z^b$ if $z\to0$. Hence, in particular, 
$\phi_0$ \emph{must be bounded on the complex $z$-plane}. 

Now, let us suppose that both $M\pm s$ are odd: $M+s=2k_1+1$ and $M-s=2k_2
+1$ for some $k_1,k_2\in\mathbb{Z}$. Then, apart from the isolated points 
$\xi_\pm,\bar\xi_\pm$, the function $F$ given by (\ref{eq:3.3.22}) is smooth 
on ${\cal S}$. Hence, to ensure the smoothness of $\phi$ on ${\cal S}-\{\xi
_\pm,\bar\xi_\pm\}$, too, the function $\phi_0$ must also be well defined and 
differentiable on the whole complex $z$-plane except possibly at $z=0,1$. 
But at the same time $\phi_0(z)$ must tend to zero both when $z\to0$ and 
when $z\to\infty$, i.e. it is \emph{bounded}. But then the next Lemma 
forces $\phi_0$ to be identically zero:
\begin{lemma*}
Let $\phi_0:\mathbb{C}-\{0\}\to\mathbb{C}$ be continuous and differentiable
on $\mathbb{C}-\{0,1\}$. If $\phi_0(z)$ tends to zero as $1/z^a$ in the 
$z\to\infty$ limit for some $a>0$, and if it tends to zero as $z^b$ in the 
$z\to0$ limit for some $b>0$, then $\phi_0$ is identically zero. 
\end{lemma*}
\begin{proof}
Let $n\in\mathbb{N}$ such that $1/n<a$. Then $\sqrt[n]{z-1}\phi_0$ tends to 
zero both in the $z\to0$ and in the $z\to\infty$ limits, and it is zero at 
$z=1$. Let ${\cal R}(n)$ be the Riemann surface associated naturally with 
the $\sqrt[n]{z-1}$ function: it is the union of $n$ copy of the complex 
$z$-plane which has been cut along the $z\geq1$ portion of the real axis, 
and in which the ${\rm Im}(z)>0$ (upper) side edge of the cut in the $k$th 
copy has been identified with the ${\rm Im}(z)<0$ (lower) side edge of the
cut in the $(k+1)$th copy for any $k\in\mathbb{N}$ modulo $n$. Then 
$\sqrt[n]{z-1}\phi_0$ extends to ${\cal R}(n)-\{0,1\}$ to be a 
differentiable function, which tends to zero in the $z\to0$, $z\to1$ and 
$z\to\infty$ limits. Thus, for any $\epsilon>0$, there exists a positive 
number, $r>0$, such that on the neighbourhood $U_r(0):=\{\xi\in\mathbb{C}
\vert\,\vert\xi\vert<r\}$ and $U_r(1):=\{\xi\in\mathbb{C}\vert\,\vert\xi-1
\vert<r\}$ of $0$ and $1$, respectively, the modulus $\vert\sqrt[n]{z-1}
\phi_0\vert$ is smaller than $\epsilon$. Also, there is a positive number, 
$R>0$, such that this modulus is smaller than $\epsilon$ on $U_R(\infty):=
\{\xi\in\mathbb{C}\vert\,\vert\xi\vert>R\}$. These neighbourhoods determine 
open neighbourhoods of $0$, $1$ and $\infty$, respectively, in ${\cal R}
(n)$, too, such that $\vert\sqrt[n]{z-1}\phi_0\vert<\epsilon$ on these
neighbourhoods. Thus, $\max\{\vert\sqrt[n]{z-1}\phi_0(z)\vert\,\vert\,\, z
\in\partial U_r(0)\cup\partial U_r(1)\cup\partial U_R(\infty)\}\leq\epsilon$. 
But $K:={\cal R}(n)-U_r(0)-U_r(1)-U_R(\infty)$ is compact and $\partial 
U_r(0)\cup\partial U_r(1)\cup\partial U_R(\infty)=\partial K$, and hence, by 
the maximum modulus principle (see e.g. \cite{Sc}), $\vert\sqrt[n]{z-1}\phi
_0(z)\vert\leq\epsilon$ for any $z\in K$, too. Therefore, $\vert\sqrt[n]{z-1}
\phi_0\vert$ must be smaller than an arbitrarily chosen small $\epsilon>0$ on 
the whole ${\cal R}(n)$, implying that $\phi_0$ is identically vanishing. 
\end{proof}
\noindent
Thus, in this case, there are no most classical states. 

Next suppose that both $M\pm s$ are even: $M+s=2k_1$ and $M-s=2k_2$. Then 
by (\ref{eq:3.3.22}) 
\begin{equation*}
F=\frac{1}{\sqrt{(\xi-\xi_+)(\bar\xi-\xi_+)(\xi-\xi_-)(\bar\xi-\xi_-)}}
\bigl(\xi-\xi_+\bigr)^{k_1}\bigl(\bar\xi-\xi_+\bigr)^{k_2}\bigl(\xi-\xi_-\bigr)
^{-k_2}\bigl(\bar\xi-\xi_-\bigr)^{-k_1}\bigl(1+\xi\bar\xi\bigr).
\end{equation*}
The first factor has jumps on straight lines parallel with the positive real 
axis that start, respectively, at $\xi_\pm,\bar\xi_\pm$. Hence, to be able to 
obtain a wave function $\phi$ that is smooth on ${\cal S}-\{\xi_\pm,\bar\xi
_\pm\}$, there should exist a smooth function $\Phi$ on ${\cal S}-\{\xi_\pm,
\bar\xi_\pm\}$ satisfying 
\begin{equation}
\Phi\sqrt{(\xi-\xi_+)(\bar\xi-\xi_+)(\xi-\xi_-)(\bar\xi-\xi_-)}=\phi_0.
\label{eq:3.3.26}
\end{equation}
By the definition (\ref{eq:3.3.17}) of $z$, this condition can be rewritten 
in the form 
\begin{equation*}
\Phi\bigl(\xi-\xi_-\bigr)\bigl(\bar\xi-\xi_+\bigr)=\sqrt{z}\phi_0=:\phi_1.
\end{equation*}
Since the expression on the left is smooth on ${\cal S}-\{\xi_\pm,\bar
\xi_\pm\}$ and the map $(\xi,\bar\xi)\mapsto z$ is surjective and smooth on 
the $\xi$-plane except at $\xi_-$ and $\bar\xi_+$, $\phi_1$ is a well 
defined differentiable function on the whole complex $z$-plane except 
possibly at $z=0,1$. By (\ref{eq:3.3.25}) $\phi_1$ is still bounded: in the 
$z\to0$ limit it tends to zero as $z^b$ with $b\geq3/2$, and although in the 
$z\to\infty$ limit it does not necessarily tend to zero, but by 
(\ref{eq:3.3.25}) it is still bounded. 

However, (\ref{eq:3.3.26}) can be rewritten in the other equivalent form 
\begin{equation*}
\Phi\bigl(\xi-\xi_+\bigr)\bigl(\bar\xi-\xi_-\bigr)=\frac{\phi_0}{\sqrt{z}}
=:\phi_2,
\end{equation*}
and the previous analysis can be repeated. We obtain that $\phi_2(z)$ tends 
to zero in the $z\to\infty$ limit as $1/z^a$ with $a\geq3/2$, while it 
remains bounded in the $z\to0$ limit. Then, however, $\phi_1\phi_2=(\phi_0)
^2$ satisfies the conditions of the Lemma above, and hence $\phi_0$ must be 
identically zero, which is a contradiction. 

Next let us suppose that $M+s=2k_1$ and $M-s=2k_2+1$. Then by 
(\ref{eq:3.3.22})
\begin{equation*}
F=\frac{1}{\sqrt{(\xi-\xi_+)(\bar\xi-\xi_-)}}\bigl(\xi-\xi_+\bigr)^{k_1}
\bigl(\bar\xi-\xi_+\bigr)^{k_2}\bigl(\xi-\xi_-\bigr)^{-k_2-1}\bigl(\bar\xi-
\xi_-\bigr)^{-k_1}\bigl(1+\xi\bar\xi\bigr).
\end{equation*}
Apart from the first factor, this is smooth and non-zero on ${\cal S}-\{
\xi_\pm,\bar\xi_\pm\}$, but the first factor is not only singular at isolated 
points, but also discontinuous on certain lines in the $\xi$-plane. Hence, 
to obtain a wave function $\phi$ that is smooth on ${\cal S}-\{\xi_\pm,\bar
\xi_\pm\}$, there should exist a function $\Phi$ which is smooth there and 
satisfies 
\begin{equation}
\Phi\sqrt{(\xi-\xi_+)(\bar\xi-\xi_-)}=\phi_0. \label{eq:3.3.27}
\end{equation}
We show that if (\ref{eq:3.3.27}) were satisfied, then there would exist a 
Riemann surface $\tilde{\cal R}$ and a complex analytic extension $\tilde
\phi_0$ of $\phi_0$ from the complex $z$-plane to $\tilde{\cal R}$ which 
would take its maximum at some point of $\tilde{\cal R}$, in contrast to the 
statement of the maximum modulus principle of complex analysis (see e.g.
\cite{Sc}). 

The factor $\sqrt{\xi-\xi_+}$ on the left of (\ref{eq:3.3.27}) has a jump 
along the straight line $\gamma_1(x):=x+{\rm i}y_+$, $x\geq x_+$, where $x_++
{\rm i}y_+:=\xi_+$ defines the real and imaginary parts of $\xi_+$. In fact, 
if $\xi_0=x_0+{\rm i}y_+$, $x_0>x_+$, is any point of $\gamma_1$ and $\xi\to
\xi^\pm_0$ denotes the limit at $\xi_0$ along any complex path through 
$\xi_0$ in the $y>y_+$/$y<y_+$ side of the line $\gamma_1$, then $\lim
_{\xi\to\xi^+_0}\sqrt{\xi-\xi_+}=(x_0-x_+)$ and $\lim_{\xi\to\xi^-_0}\sqrt{\xi-
\xi_+}=-(x_0-x_+)$. Since $\Phi$ is continuous and $\sqrt{\bar\xi-\xi_-}$ is 
continuous in a neighbourhood of $\gamma_1$, by (\ref{eq:3.3.27}) the 
function $\phi_0$ must be discontinuous along $\gamma_1$: 
\begin{equation*}
\lim_{\xi\to\xi^+_0}\phi_0\bigl(z(\xi,\bar\xi)\bigr)=-\lim_{\xi\to\xi^-_0}\phi_0
\bigl(z(\xi,\bar\xi)\bigr).
\end{equation*}
Analogously, the factor $\sqrt{\bar\xi-\xi_-}$ in (\ref{eq:3.3.27}) has a 
similar discontinuity at any point $\xi_0=x_0-{\rm i}y_-$, $x_0>x_-$, of the 
straight line $\gamma_2(x):=x-{\rm i}y_-$, $x\geq x_-$. Thus $\phi_0$ has an 
analogous jump along the straight line $\gamma_2$, too. 

Taking the $\xi$ and $\bar\xi$ derivative of (\ref{eq:3.3.27}) we find, 
respectively, that 
\begin{eqnarray*}
&{}&\sqrt{\xi-\xi_+}\sqrt{\bar\xi-\xi_-}\Bigl(\frac{\partial\Phi}
  {\partial\xi}+\frac{1}{2}\frac{\Phi}{(\xi-\xi_+)(\bar\xi-\xi_-)}\Bigr)=
  \phi^\prime_0z\frac{\xi_--\xi_+}{(\xi-\xi_+)(\xi-\xi_-)}, \\
&{}&\sqrt{\xi-\xi_+}\sqrt{\bar\xi-\xi_-}\Bigl(\frac{\partial\Phi}{\partial
  \bar\xi}+\frac{1}{2}\frac{\Phi}{(\xi-\xi_+)(\bar\xi-\xi_-)}\Bigr)=
  \phi^\prime_0z\frac{\xi_+-\xi_-}{(\bar\xi-\xi_+)(\bar\xi-\xi_-)},
\end{eqnarray*}
where $\phi^\prime_0$ denotes the derivative of $\phi_0$ with respect to $z$. 
The coefficients of $\sqrt{\xi-\xi_+}$ on the left of these equations are 
continuous along the straight line $\gamma_1$ (except at $\xi_+$), and these 
left hand sides have well defined limits from both the $y>y_+$ and the 
$y<y_+$ sides of $\gamma_1$. Since, on the right of these equations, the 
coefficients of $\phi^\prime_0$ are continuous (except at $\xi_+$), this 
implies, first, that $\phi^\prime_0$ has a well defined limit at the image 
under $z$ of any point of $\gamma_1$ from both sides of $\gamma_1$, and, 
second, that these limits are $-1$-times of each other. Taking higher 
derivatives of (\ref{eq:3.3.27}), an analogous argumentation gives that any 
$n$th order derivative $\phi^{(n)}_0$ of $\phi_0$ behaves in the same way: 
\begin{equation}
\lim_{\xi\to\xi^+_0}\phi^{(n)}_0\bigl(z(\xi,\bar\xi)\bigr)=-\lim_{\xi\to\xi^-_0}
\phi^{(n)}_0\bigl(z(\xi,\bar\xi)\bigr). \label{eq:3.3.28}
\end{equation}
$\phi^{(n)}_0$ has an analogous jump along the straight line $\gamma_2$, too. 

Let ${\cal R}_1$ denote the complex $\xi$-plane which is cut along the 
straight lines $\gamma_1$ and $\gamma_2$, and let ${\cal R}_2$ denote another 
copy of ${\cal R}_1$. Then we form the Riemann surface ${\cal R}$ by 
identifying the $y>y_+$ side edge along $\gamma_1$ in ${\cal R}_1$ with the 
$y<y_+$ side edge along $\gamma_1$ in ${\cal R}_2$; and the $y<y_+$ side edge 
along $\gamma_1$ in ${\cal R}_1$ with the $y>y_+$ side edge along $\gamma_1$ 
in ${\cal R}_2$. We make the analogous identification of the edges along the 
other straight line $\gamma_2$, too. $\xi_+$ and $\bar\xi_-$ will be the two 
branch points in ${\cal R}$. Then we define $\Phi$ on the second copy 
${\cal R}_2$ to be equal to $\Phi$ on ${\cal R}_1$, and $\phi_0$ (or, more 
precisely, $\phi_0\circ z$) on ${\cal R}_2$ to be $-1$ times of $\phi_0$ on 
${\cal R}_1$. With these definitions we have extended $\Phi$ and $\phi_0$ 
from ${\cal R}_1$ to ${\cal R}$, and by (\ref{eq:3.3.28}) these are smooth 
and (\ref{eq:3.3.27}) still holds on ${\cal R}$. 

Finally, let $\tilde{\cal R}_1$ and $\tilde{\cal R}_2$ be two copies of the 
complex $z$-plane which have been cut along the image of the straight lines 
$\gamma_1$ and $\gamma_2$, and then construct the Riemann surface $\tilde
{\cal R}$ from these two in the way analogous to how ${\cal R}$ was 
constructed. This $\tilde{\cal R}$ can be considered as the image of ${\cal 
R}$ under the map $z$ in a wider sense. In fact, under the action of the map
$z$, the branch points $\xi_+$ and $\bar\xi_-$ of ${\cal R}$ are pushed out 
to infinity of $\tilde{\cal R}$, and the branch point $1$ of $\tilde{\cal R}$ 
corresponds to the whole real axis and the infinity of ${\cal R}$. Thus, 
${\cal R}$ and $\tilde{\cal R}$ are \emph{not} homeomorphic to each other. 
In particular, while ${\cal R}$ is connected, $\tilde{\cal R}$ consists of 
two disconnected components $\tilde{\cal R}^\prime$ and $\tilde{\cal R}
^{\prime\prime}$, each of them being homeomorphic to the cylinder $S^1\times
\mathbb{R}$, and which are `connected' to each other only at infinity. 

The function $\phi_0$ on ${\cal R}$ defines a continuous function $\tilde
\phi_0$ both on $\tilde{\cal R}^\prime-\{0\}$ and $\tilde{\cal R}^{\prime\prime}
-\{0\}$, and this $\tilde\phi_0$ is differentiable on $\tilde{\cal R}^\prime-
\{0,1\}$ and $\tilde{\cal R}^{\prime\prime}-\{0,1\}$. Moreover, in the $z\to0$ 
limit $\tilde\phi_0(z)$ tends to zero as $z^b$ with $b\geq1/2$, and in the 
$z\to\infty$ limit it tends to zero as $1/z^a$ with $a\geq1/2$. Then, 
repeating the argumentation of the proof of the Lemma above, we conclude 
that $\tilde\phi_0$, and hence $\phi_0$ itself, must be identically vanishing. 

The previous analysis can be repeated if $M+s$ is odd and $M-s$ is even 
with the same conclusion, and hence there are no most classical states in 
this case either.

\bigskip
\noindent
$\bullet$ {\bf The non-existence of differentiable wave functions: The 
exceptional case}

\noindent
The logic of the proof of the non-existence of differentiable wave functions 
in the exceptional case is similar to that in the generic case. Thus we only 
sketch the key steps. 

First we show that $(\xi,\bar\xi)\mapsto u$ is a surjective $\mathbb{C}\to
\mathbb{C}$ map, and hence the domain of $\phi_0$ should be the entire 
complex $u$-plane except possibly curves or finitely many isolated points. 
By (\ref{eq:3.3.20}) $u:=u_1-u_2=0$ on the real axis $\bar\xi=\xi$ of the 
complex $\xi$-plane, and hence we should only show the existence of a 
pre-image of any $u\not=0$. However, with the notation $\xi=:r\exp({\rm i}
\chi)$ the definition of $u$ gives 
\begin{eqnarray*}
&{}&r^2+r\Bigl(\exp({\rm i}\chi)\bigl({\rm i}+\frac{1}{u}\bigr)+\exp(
  -{\rm i}\chi)\bigl({\rm i}-\frac{1}{u}\bigr)\Bigr)-1=0, \\
&{}&r^2+r\Bigl(\exp(-{\rm i}\chi)\bigl(-{\rm i}+\frac{1}{\bar u}\bigr)+\exp(
  {\rm i}\chi)\bigl(-{\rm i}-\frac{1}{\bar u}\bigr)\Bigr)-1=0;
\end{eqnarray*}
implying that 
\begin{eqnarray*}
&{}&\exp(2{\rm i}\chi)=\frac{u+\bar u-2{\rm i}u\bar u}{u+\bar u+2{\rm i}u
  \bar u}, \\
&{}&r^2+\frac{1}{2}r\Bigl(\exp({\rm i}\chi)\bigl(\frac{1}{u}-\frac{1}{\bar u}
  \bigr)+\exp(-{\rm i}\chi)\bigl(\frac{1}{\bar u}-\frac{1}{u}\bigr)\Bigr)
  -1=0.
\end{eqnarray*}
These show that $(\xi,\bar\xi)\mapsto u$ is, in fact, a surjective 
$\mathbb{C}\to\mathbb{C}$ map. This maps the points $\xi=\pm{\rm i}$ of the 
complex $\xi$-plane into infinity of the complex $u$-plane. 

If the (complex) eigenvalue $C$ in (\ref{eq:3.3.21}) were non-zero, then 
solution (\ref{eq:3.3.21}) would have an essential singularity at both 
$\xi={\rm i}$ and $-{\rm i}$, and with these singularities $\phi$ could not 
be locally square integrable. Therefore, $C$ must be vanishing and the local 
solution is 
\begin{equation}
\phi=\phi_0\bigl(1+\xi\bar\xi\bigr)\bigl(\xi+{\rm i}\bigr)^{s-1}\bigl(
\bar\xi+{\rm i}\bigr)^{-s-1}. \label{eq:3.3.29}
\end{equation}
The coefficient of $\phi_0$ in (\ref{eq:3.3.29}), 
\begin{equation*}
F:=\bigl(1+\xi\bar\xi\bigr)\bigl(\xi+{\rm i}\bigr)^{s-1}\bigl(\bar\xi+{\rm i}
\bigr)^{-s-1},
\end{equation*}
is bounded in the limit $\xi\to\infty$, and it is non-singular and non-zero 
everywhere except at the two points $\xi=\pm{\rm i}$. In the limit $\xi\to-
{\rm i}$ the function $\phi_0$ depends essentially only on $\xi+{\rm i}$, 
and hence, in this limit, we may write $\phi_0\simeq\alpha(\xi+{\rm i})^a$ 
for some $a=a_1+{\rm i}a_2\in\mathbb{C}$, $a_1,a_2\in\mathbb{R}$, and $\alpha
\in\mathbb{C}$. Similarly, in the limit $\xi\to{\rm i}$, we have that $\phi_0
\simeq\alpha(\bar\xi+{\rm i})^b$, $b=b_1+{\rm i}b_2$ for some $b_1,b_2\in
\mathbb{R}$. 

On a small enough open neighbourhood $U_R(-{\rm i})$ of $-{\rm i}$ we write 
$\xi=-{\rm i}+r\exp({\rm i}\chi)$, and, on this neighbourhood, 
\begin{equation}
\phi(\xi+{\rm i})^{s-1}\simeq\alpha(\xi+{\rm i})^{a+s-1}=\alpha r^{a_1+s-1}
\exp(-a_2\chi)\exp\Bigl({\rm i}\bigl((a_1+s-1)\chi+a_2\ln r\bigr)\Bigr).
\label{eq:3.3.30}
\end{equation}
Hence, repeating the argumentation that we had in the generic case, we find 
that the condition of the square integrability of $\phi$ yields that $a_1+s
>0$. Also, the square integrability of $\phi$ on the neighbourhood $U_R
({\rm i})$ of ${\rm i}$ yields $b_1-s>0$. 

Since $u\to\infty$ in both limits $\xi\to\pm{\rm i}$, the asymptotic form 
of $\phi_0$ in both these limits are the same. Hence, $a=b$, yielding that 
$a_1=b_1>0$, i.e. $\phi_0$ must tend to zero if $\xi\to\pm{\rm i}$, i.e. if 
$u\to\infty$. 

If, in addition, the wave function $\phi$ is required to be well defined on 
the whole ${\cal S}$ except only at $\xi=\pm{\rm i}$, then we find that $a$ 
and $b$ must be real and $a=-s+1+n_1$, $b=s+1+n_2$, where $n_1,n_2=0,1,2,...$. 
Then by the equality $a=b$ these yield that the two integers are not 
independent: $2s=n_1-n_2$. Therefore, to summarize, $\phi_0(u)$ must tend to 
zero in the $u\to\infty$ limit as $1/u^a$, where $a=1+\frac{1}{2}(n_1+n_2)$ 
and $n_1,n_2=0,1,2,...$ such that $2s=n_1-n_2$. In particular, \emph{it must 
be bounded on the complex $u$-plane}. 

Now let us suppose that $s$ is an integer. Then, apart from the points $\pm
{\rm i}$, the function $F$ is smooth on ${\cal S}$. Hence, if we want the 
wave function $\phi$ to be differentiable on ${\cal S}-\{\pm{\rm i}\}$, then 
$\phi_0$ have to be well defined on the $\xi$-plane minus the points 
$\pm{\rm i}$. Therefore, $\phi_0$ must be a differentiable bounded function 
on the \emph{entire} $u$-plane. Hence, by Liouville's theorem, it must be 
constant. But since $\phi_0$ must be vanishing at $u=\infty$, it would have 
to be identically zero. This is a contradiction. 

If $s$ is half odd integer, i.e. $s=S+1/2$ for some $S\in\mathbb{Z}$, 
then 
\begin{equation*}
F=\frac{1}{\sqrt{(\xi+{\rm i})(\bar\xi+{\rm i})}}\bigl(\xi+{\rm i}\bigr)^S
\bigl(\bar\xi+{\rm i}\bigr)^{-S-1}\bigl(1+\xi\bar\xi\bigr).
\end{equation*}
Again, as in the generic case, the existence of a wave function $\phi$ that 
is differentiable on ${\cal S}-\{\pm{\rm i}\}$ is equivalent to the existence 
of a differentiable function $\Phi$ on ${\cal S}-\{\pm{\rm i}\}$ satisfying 
\begin{equation*}
\Phi\sqrt{(\xi+{\rm i})(\bar\xi+{\rm i})}=\phi_0.
\end{equation*}
However, by the definition of $u$ we can rewrite this condition in the form 
\begin{equation*}
\Phi\sqrt{\bar\xi-\xi}=\sqrt{u}\phi_0.
\end{equation*}
On the upper half of the complex $\xi$-plane, where $y:={\rm Im}(\xi)>0$, 
we have that $\sqrt{\bar\xi-\xi}=-(1+{\rm i})\sqrt{y}$, while on the lower 
half plane that $\sqrt{\bar\xi-\xi}=(1+{\rm i})\sqrt{\vert y\vert}$. Hence, 
$\Phi\sqrt{\bar\xi-\xi}$ is continuous on the whole complex $\xi$-plane, 
and it is smooth everywhere outside the real axis $\bar\xi=\xi$. Since, 
however, the image of this real axis in the $u$-plane is $0$, we conclude 
that $\sqrt{u}\phi_0$ is differentiable on the $u$-plane except its origin. 
Clearly, $\sqrt{u}\phi_0$ is vanishing at $u=0$, and, in the $u\to\infty$
limit, it falls off as $(1/u)^{(1+n_1+n_2)/2}$. In particular, it is 
bounded on the whole $u$-plane. But then the Lemma above implies that 
$\sqrt{u}\phi_0$ is identically zero, yielding that $\phi_0$ itself is 
identically vanishing. This is a contradiction. 

Therefore, \emph{there are no differentiable, normalizable most classical
states neither in the generic nor in the exceptional case}. 


\section{Summary and conclusions}
\label{sec-4}

It is well known (see e.g. \cite{Kl,ZFG}) that in a Heisenberg system, 
generated by $\{{\bf q},{\bf p},{\bf I}\}$, (1) the most classical states 
always exist, and (2) they can be parameterized by the expectation values, 
$\langle{\bf q}\rangle_\phi$ and $\langle{\bf p}\rangle_\phi$, and by the 
quotient of the standard deviations, $\lambda=\Delta_\phi{\bf p}/\Delta_\phi
{\bf q}$. (3) There is no restriction on the expectation values, they can 
take any value; and hence (4) the set of the pairs $(\langle{\bf q}\rangle
_\phi,\langle{\bf p}\rangle_\phi)$ is in a one-to-one correspondence with the 
points of the phase space $T^*\mathbb{R}\simeq\mathbb{R}^2$ of the 
`corresponding' \emph{classical} system. (5) Although the quotient $\Delta
_\phi{\bf p}/\Delta_\phi{\bf q}$ can be any positive number, the product of 
the standard deviations is a \emph{universal constant}: $\Delta_\phi{\bf p}\,
\Delta_\phi{\bf q}=\hbar/2$. Thus, the most classical states are 
the \emph{minimal uncertainty states} (i.e. in which the function $F:\phi
\mapsto\Delta_\phi{\bf p}\Delta_\phi{\bf q}$ on the unit sphere in the Hilbert 
space of the states takes its \emph{minimum}), too. 

As a result of the present investigations, we found that $E(3)$ 
invariant elementary quantum mechanical systems provide pairs of observables 
such that, for the most classical states with respect to them, \emph{each of 
the above properties is failed to be satisfied} in some of these examples. 
In particular: 
(1) most classical states exist for the pair $({\bf p}(\alpha),{\bf J}
(\beta))$ only when the direction $\alpha^i$ is orthogonal to $\beta^i$, and 
for $({\bf p}(\alpha),{\bf C}(\beta))$ only when the angle between these 
directions is zero or an acute angle, but no such state (which could be 
represented by a \emph{differentiable} wave function) exists at all for 
$({\bf C}(\alpha),{\bf C}(\beta))$. 
(2) The most classical states, when they exist, depend not only on the 
expectation values (or the parameters fixing the expectation values), but 
also on an (almost) completely free functions, $A(\theta)$, $\phi_0(w)$ and 
$A(\varphi)$ in the $({\bf p}(\alpha),{\bf J}(\beta))$, $({\bf J}(\alpha),
{\bf J}(\beta))$ and $({\bf p}(\alpha),{\bf C}(\beta))$ cases, respectively, 
i.e. on \emph{infinitely many parameters}, too. 
(3) The expectation values can be considerably restricted: in the $({\bf p}
(\alpha),{\bf J}(\beta))$ and $({\bf p}(\alpha),{\bf C}(\beta))$ cases, 
respectively, the expectation value of the linear momentum must be zero, 
$\langle{\bf p}(\alpha)\rangle_\phi=0$, and restricted by $\vert\langle
{\bf p}(\alpha)\rangle_\phi\vert\leq P\alpha_i\beta^i$; hence 
(4) these values do not exhaust all their classically allowed values. 
(5) Although, as in the case of the Heisenberg system, the quotient of the 
standard deviations of the two observables (i.e. the range of $\lambda$) is 
not restricted at all, their product depends on $\lambda$ without a minimum 
at any finite value of $\lambda$. Thus, there are \emph{no} minimal 
uncertainty states when the most classical states exist. (This last property 
has already been realized in \cite{Aretal76} in the special case of the most 
classical states for the two angular momentum operators $({\bf J}_1,{\bf J}
_2)$ in the $su(2)$ algebra.) However, when no most classical states exist, 
then, in principle, minimal uncertainty states might exist, but these could 
not saturate the uncertainty relation: in these the \emph{strict inequality} 
would have to hold. 

These results show explicitly that the most classical states are associated 
\emph{only with pairs of observables}, rather than with the whole 
\emph{algebra} of its (basic) observables: for a certain choice of a pair 
of observables from this algebra the most classical states exist, but for 
a different choice they do not at all. 

It is known that, in simultaneous measurements of conjugate 
(non-commuting) observables of the Heisenberg system \cite{ArKe}, the final 
states after the \emph{most accurate} (or ideal) measurements are the most 
classical (pure), rather than more general mixed states \cite{SHe}. Hence, 
assuming that this link between the ideal simultaneous measurements and the 
most classical states is characteristic not only to Heisenberg systems, the 
lack of the most classical states for a pair of observables indicates that 
the final state of the system after the simultaneous measurement of these 
observables should \emph{necessarily} be \emph{mixed}. 


\section{Acknowledgments}

Thanks are due to the referee for his critical remarks and for 
several useful references. 

This research did not receive any specific grant from funding agencies in 
the public, commercial, or not-for-profit sectors. The author has no 
conflicts to disclose.

\appendix

\section{Appendix}
\label{sec-A}

In Appendix \ref{sub-A.1}, mostly to fix the notations, we summarize the 
necessary differential geometric background that we need both in the main 
part of the paper and in Appendix \ref{sub-A.2}. Then, in Appendix 
\ref{sub-A.2}, we present the unitary, irreducible representations of the 
Euclidean group $E(3)$ in the form that we use. As far as we can see, it is 
these representations that fit most naturally to the Euclidean group; and 
they are spanned \emph{precisely} by the spin weighted spherical harmonics. 
In deriving these, we use the analogous representations of $SU(2)$, given in 
Appendix A.1.4 of \cite{Sz1} (where a more detailed quantum mechanical 
interpretation of the geometric notions is also given). These representations 
are analogous in their spirit to those of the Poincar\'e group summarized in
\cite{StWi}.

\subsection{Complex coordinates and the edth operators on ${\cal S}$}
\label{sub-A.1}

Let $p^i$, $i=1,2,3$, denote Cartesian coordinates in $\mathbb{R}^3$ and 
${\cal S}:=\{p^i\in\mathbb{R}^3\vert P^2:=\delta_{ij}p^ip^j={\rm const}\}$, 
the metric 2-sphere of radius $P$. The complex stereographic coordinates, 
projected from the \emph{north} pole, are defined on $U_n:={\cal S}-\{
(0,0,P)\}$, the sphere minus its \emph{north} pole, by $\zeta:=\exp({\rm i}
\varphi)\cot(\theta/2)$, where $(\theta,\varphi)$ are the standard spherical 
polar coordinates. In terms of $(\zeta,\bar\zeta)$, the Cartesian coordinates 
of the point $p^i\in U_n$ are 
\begin{equation}
p^i=P\Bigl(\frac{\bar\zeta+\zeta}{1+\zeta\bar\zeta},{\rm i}\frac{\bar\zeta-
\zeta}{1+\zeta\bar\zeta},\frac{\zeta\bar\zeta-1}{1+\zeta\bar\zeta}\Bigr). 
\label{eq:A.1.2}
\end{equation}
The outward pointing unit normal to ${\cal S}$ at the point $p^i$ is $n^i:=
p^i/P$. This normal is completed to be a basis by the complex vector field 
\begin{equation}
m^i:=\frac{1}{\sqrt{2}}\Bigl(\frac{1-\zeta^2}{1+\zeta\bar\zeta},
  {\rm i}\frac{1+\zeta^2}{1+\zeta\bar\zeta},\frac{2\zeta}{1+\zeta\bar\zeta}
\Bigr) \label{eq:A.1.3}
\end{equation}
and its complex conjugate $\bar m^i$. These are orthogonal to $n^i$, null 
(i.e. $m^im_i=0$), normalized with respect to each other (i.e. $m^i\bar m_i
=1$), and $p^im^j\bar m^k\varepsilon_{ijk}={\rm i}P$ holds. (Recall that, in 
the present paper, the metric on $\mathbb{R}^3$ is chosen to be the 
\emph{positive definite} $\delta_{ij}$, rather than the negative definite 
spatial part of the Minkowski metric $\eta_{ab}:={\rm diag}(1,-1,-1,-1)$, 
where $a,b=0,i$.) The vector field $m^i$, as a differential operator, 
is given by 
\begin{equation}
m^i\bigl(\frac{\partial}{\partial p^i}\bigr)=\frac{1}{\sqrt{2}P}\bigl(
1+\zeta\bar\zeta\bigr)\Bigl(\frac{\partial}{\partial\bar\zeta}\Bigr). 
\label{eq:A.1.3m}
\end{equation}
Hence, $\zeta$ is a local \emph{anti-holomorphic} coordinate on $U_n$. Also 
in these coordinates, the line element of the metric and the corresponding 
area element on ${\cal S}$ of radius $P$, respectively, are 
\begin{equation}
dh^2=\frac{4P^2}{(1+\zeta\bar\zeta)^2}d\zeta d\bar\zeta=P^2\bigl(d\theta^2+
\sin^2\theta\,d\varphi^2\bigr), \hskip 10pt 
{\rm d}{\cal S}=\frac{-2{\rm i}P^2}{(1+\zeta\bar\zeta)^2}d\zeta\wedge d
\bar\zeta=P^2\sin\theta\,d\theta\wedge d\varphi. \label{eq:A.1.4}
\end{equation}
There are analogous constructions on $U_s:={\cal S}-\{(0,0,-P)\}$, on the 
2-sphere minus the \emph{south} pole, too; and the structures defined on 
$U_n$ are related to those introduced on $U_s$ smoothly on the overlap $U_n
\cap U_s$. 

Considering $\mathbb{R}^3$ to be the $p^0=P$ hyperplane of the Minkowski 
space $\mathbb{R}^{1,3}$ with the Cartesian coordinates $p^a=(p^0,p^i)$ 
and the flat metric $\eta_{ab}$, the 2-sphere ${\cal S}$ is just the 
intersection of the $p^0=P$ hyperplane with the null cone of the origin $p
^a=0$ in $\mathbb{R}^{1,3}$. Hence, for any $p^i\in{\cal S}$, there is a 
spinor $\pi^A$, the `spinor constituent' of $p^i$, such that $p^i=\sigma
^i_{AA'}\pi^A\bar\pi^{A'}$, and $\pi^A$ is unique only up to the phase 
ambiguity $\pi^A\mapsto\exp({\rm i}\gamma)\pi^A$, $\gamma\in[0,2\pi)$. 
(Here $\sigma^a_{AA'}$ are the standard $SL(2,\mathbb{C})$ Pauli matrices
(including the factor $1/\sqrt{2}$), according to the conventions of
\cite{PR}. Note, however, that we lower and raise the small Latin indices
by $\delta_{ij}$ and its inverse. The capital Latin indices $A,B,...=0,1$ 
are concrete spinor name indices, referring to some constant normalized 
spin frame in $\mathbb{R}^3$. The spinor name indices are raised and lowered 
by the standard symplectic metric $\varepsilon^{AB}$ and its inverse 
\cite{PR,HT}.) Since $P^2=\delta_{ij}p^ip^j=-\eta_{ij}\sigma^i_{AA'}\sigma^j
_{BB'}\pi^A\bar\pi^{A'}\pi^B\bar\pi^{B'}=-(\eta_{ab}\sigma^a_{AA'}\sigma^b_{BB'}-
\sigma^0_{AA'}\sigma^0_{BB'})\pi^A\bar\pi^{A'}\pi^B\bar\pi^{B'}=(\sigma^0_{AA'}
\pi^A\bar\pi^{A'})^2$, the norm of $\pi^A$ with respect to $\sqrt{2}\sigma
^0_{AA'}$ is $(\sqrt{2}P)^{1/2}$. However, $\pi^A$ as a spinor \emph{field} 
is well defined only on ${\cal S}$ \emph{minus one point} (see e.g. 
\cite{HT,EaTod}). Thus $\pi^A$ on $U_n$ and the analogous one on $U_s$ are 
only \emph{locally defined} `spinorial coordinates' on ${\cal S}$. 

The complex line bundles ${\cal O}(-2s)$ over ${\cal S}$ can be introduced 
using the totally symmetric N-type spinor fields of rank $2\vert s\vert$ on 
${\cal S}$: if $s=-\vert s\vert\leq0$ then these spinor fields are 
\emph{unprimed} and their principal spinor at the point $p^i=\sigma^i_{AA'}
\pi^A\bar\pi^{A'}$ is $\pi^A$; and if $s=\vert s\vert>0$ then the spinor 
fields are \emph{primed} and their principal spinor is $\bar\pi^{A'}$. 
(Recall that e.g. $\lambda^A$ is called a $2\vert s\vert$-fold principal 
spinor of the totally symmetric spinor $\phi^{A_1...A_{2\vert s\vert}}$ if $\phi
^{A_1...A_{2\vert s\vert}}\lambda_{A_1}=0$ holds, in which case $\phi
^{A_1...A_{2\vert s\vert}}$ necessarily has the form $\phi\lambda^{A_1}\cdots\lambda
^{A_{2\vert s\vert}}$ for some $\phi$. These spinors are called null or of type 
N, see e.g. \cite{PR,HT}.) Hence, e.g. on the domain $U_n$, these spinor 
fields have the form $\phi^{A_1...A_{2\vert s\vert}}=\phi\pi^{A_1}...\pi
^{A_{2\vert s\vert}}$ and $\chi^{A'_1...A'_{2\vert s\vert}}=\chi\bar\pi^{A'_1}...\bar\pi
^{A'_{2\vert s\vert}}$, where $\phi$ and $\chi$ are complex functions on $U_n$. 
Thus, the fibers of these bundles are one complex dimensional, and the 
line bundle ${\cal O}(-2s)$ is just the abstract bundle of these fibers 
over ${\cal S}$. $U_n$ and $U_s$ are local trivialization domains of 
${\cal O}(-2s)$, and the functions $\phi$ for $s=-\vert s\vert$ (and $\chi$ 
for $s=\vert s\vert$) are \emph{local} cross sections of ${\cal O}(-2s)$ on 
$U_n$. ${\cal O}(-2s)$ is globally trivializable precisely when $s=0$. 

The phase ambiguity $\pi^A\mapsto\exp({\rm i}\gamma)\pi^A$ in the principal 
spinor yields the ambiguity $\phi\mapsto\exp(-2{\rm i}\vert s\vert\gamma)
\phi$, where $\gamma$ is an arbitrary $[0,2\pi)$-valued locally defined 
function. The analogous ambiguity in the function $\chi$ is $\chi\mapsto
\exp(2{\rm i}s\gamma)\chi$. Therefore, despite this ambiguity, the
Hermitian scalar product of any two cross sections, representing 
e.g. $\phi^{A_1...A_{2\vert s\vert}}$ and $\psi^{A_1...A_{2\vert s\vert}}$ and given by 
\begin{equation}
\langle\phi^{A_1...A_{2\vert s\vert}},\psi^{A_1...A_{2\vert s\vert}}\rangle_s
:=\int_{\cal S}\bar\phi\psi{\rm d}{\cal S}, \label{eq:A.1.5}
\end{equation}
is well defined. The space of the square-integrable cross sections of 
${\cal O}(-2s)$ is a Hilbert space, and is also denoted by ${\cal H}_s$. One 
can show that this scalar product is just $(\sqrt{2}P)^{-2\vert s\vert}$ times 
the familiar, standard $L_2$-scalar product of the two spinor fields (see 
\cite{StWi}). In quantum mechanics (see Appendix \ref{sub-A.2}), these 
square integrable totally symmetric spinor fields, or, equivalently, the 
corresponding cross sections of the line bundle ${\cal O}(-2s)$, play the 
role of the wave functions, while their domain, ${\cal S}$, is analogous to 
the mass-shell in the momentum space. 

On $U_n$, we introduce the spinor field $o^A:=(\sqrt{2}P)^{-1/2}\pi^A$, 
which is completed by a spinor field $\iota^A$ on $U_n$ to be the 
Newman--Penrose spinor basis $\{o^A,\iota^A\}$ such that $o_A\iota^A=1$, 
$m^i\sigma^{AA'}_i=-o^A\bar\iota^{A'}$ and $\bar m^i\sigma^{AA'}_i=-\iota^A
\bar o^{A'}$. Then it is easy to see that $p^i\sigma^{AA'}_i=P(\iota^A\bar
\iota^{A'}-o^A\bar o^{A'})/\sqrt{2}$ also holds. Recalling that a scalar 
$\phi$ is said to have the spin weight $\frac{1}{2}(p-q)$ if under the 
rescaling $\{o^A,\iota^A\}\mapsto\{\lambda o^A,\lambda^{-1}\iota^A\}$, 
where $\lambda$ is any nowhere vanishing complex function on the domain 
of the spin frame, the scalar $\phi$ transforms as $\phi\mapsto\lambda^p
\bar\lambda^q\phi$ (see e.g. \cite{PR,HT}). Thus, ${\cal O}(-2s)$ is just 
the bundle of spin weighted scalars of weight $s$ on ${\cal S}$. In 
particular, the components of the vectors $m^i$ and $\bar m^i$ are of type 
$(1,-1)$ and $(-1,1)$, respectively, while those of $p^i$ are sums of a 
$(1,1)$ and a $(-1,-1)$ type scalar. Thus, the spin weight of $m^i$, 
$\bar m^i$ and $p^i$ is $1$, $-1$ and $0$, respectively. 

If $\delta_i$ denotes the (Cartesian components of the) covariant derivative 
operator of the induced Levi-Civita connection acting on the spinor fields 
on ${\cal S}$, then e.g. for $s=-\vert s\vert$ the $\iota^A$-spinor 
components of the covariant directional derivatives of the spinor fields, 
defined by ${\edth}_s\phi:=m^i\delta_i((\sqrt{2}P)^{-\vert s\vert}\phi\pi_{A_1}
\cdots\pi_{A_{2\vert s\vert}})\iota^{A_1}\cdots\iota^{A_{2\vert s\vert}}$ and ${\edth}'
_s\phi:=\bar m^i\delta_i((\sqrt{2}P)^{-\vert s\vert}\phi\pi_{A_1}\cdots\pi
_{A_{2\vert s\vert}})\iota^{A_1}\cdots\iota^{A_{2\vert s\vert}}$, give just the edth 
and edth-prime operators of Newman and Penrose \cite{NP} acting on the 
appropriate line bundles (see also \cite{PR,HT,EaTod,Goetal}). ${\edth}_s$ 
and ${\edth}'_s$ acting on cross sections of ${\cal O}(-2s)$ for $s=\vert s
\vert$ are defined analogously. ${\edth}_s$ increases, and ${\edth}'_s$ 
decreases the spin weight by one. In the complex stereographic coordinates 
on $U_n$, the explicit form of these operators, acting on a function $\phi$ 
of spin weight $s$, is 
\begin{equation}
{\edth}_s\phi=\frac{1}{\sqrt{2}P}\Bigl(\bigl(1+\zeta\bar\zeta\bigr)
\frac{\partial\phi}{\partial\bar\zeta}+s\zeta\phi\Bigr), \hskip 20pt
{\edth}'_s\phi=\frac{1}{\sqrt{2}P}\Bigl(\bigl(1+\zeta\bar\zeta\bigr)
\frac{\partial\phi}{\partial\zeta}-s\bar\zeta\phi\Bigr); \label{eq:A.1.7}
\end{equation}
while, in the more familiar polar coordinates $(\theta,\varphi)$, these 
operators take the form 
\begin{equation}
{\edth}_s\phi=-\frac{1}{\sqrt{2}P}e^{{\rm i}\varphi}\Bigl(\frac{\partial\phi}
{\partial\theta}-\frac{\rm i}{\sin\theta}\frac{\partial\phi}{\partial\varphi}
-s\cot\frac{\theta}{2}\phi\Bigr), \hskip 15pt
{\edth}'_s\phi=-\frac{1}{\sqrt{2}P}e^{-{\rm i}\varphi}\Bigl(\frac{\partial\phi}
{\partial\theta}+\frac{\rm i}{\sin\theta}\frac{\partial\phi}{\partial\varphi}
+s\cot\frac{\theta}{2}\phi\Bigr). \label{eq:A.1.8}
\end{equation}
If no confusion arises, simply we write ${\edth}$ and ${\edth}'$ instead of 
${\edth}_s$ and ${\edth}'_s$. These operators link the spinors $o^A$ and 
$\iota^A$: ${\edth}o^A=0$, ${\edth}'o^A=\iota^A/(\sqrt{2}P)$, ${\edth}\iota^A
=-o^A/(\sqrt{2}P)$ and ${\edth}'\iota^A=0$; which imply ${\edth}p^i=m^i$, 
${\edth}m^i=0$ and ${\edth}\bar m^i=-p^i/P^2$. 


\subsection{The unitary, irreducible representations of $E(3)$}
\label{sub-A.2}

We \emph{a priori} assume that the spectrum of the operator ${\bf p}^i$ is 
$\mathbb{R}^3$, the classical momentum space, and we endow it with the 
3-metric $\delta_{ij}$. Thus $(\mathbb{R}^3,\delta_{ij})$ is a \emph{flat} 
Riemannian 3-manifold with the globally defined Cartesian coordinates $p^i$ 
and the physically distinguished origin $p^i=0$. The action of $SU(2)$ is 
given by $p^i\mapsto\Lambda^i{}_j(A)p^j$, were $\Lambda^i{}_j(A):=-\sigma^i
_{AA'}A^A{}_B\bar A^{A'}{}_{B'}\sigma^{BB'}_j$, in which $A^A{}_B\in SU(2)$ and 
the over-bar denotes complex conjugation. (The $(-)$ sign in the expression 
of $\Lambda^i{}_j(A)$ is due to our convention that, in the present paper, 
we lower and raise the small Latin indices by the \emph{positive definite} 
$\delta_{ij}$ and its inverse.) The surfaces of transitivity of $SU(2)$ are 
${\cal S}:=\{p^i\in{\cal M}\vert P^2:=\delta_{ij}p^ip^j={\rm const}\}$ (the 
`kinetic energy-shell', being analogous to the mass-shell in the 
representation theory of the Poincar\'e group), which are 2-spheres for 
$P>0$ and the single point $p^i=0$ for $P=0$. 

The irreducible representations of $E(3)$ are labeled by the value of the 
two Casimir operators, $P$ and $w$, and hence there are two disjoint cases: 
when $P>0$ and when $P=0$. Thus, first let us suppose that $P>0$, and fix a 
point $\mathring{p}^i\in{\cal S}$. According to the method of induced 
representations, first we should find the representations of the stabilizer 
subgroup for $\mathring{p}^i$ in $E(3)$. This is $U(1)\subset SU(2)$, and, 
by Schur's lemma, all of its irreducible representations are one-dimensional, 
and these are labeled by $s=0,\pm\frac{1}{2},\pm1,...$. If $\mathring{\pi}
{}^A$ is the spinor constituent of $\mathring{p}^i$, then this 
one-dimensional representation space is chosen to be spanned by the spinor 
of the form $\mathring{\pi}{}^{A_1}\cdots\mathring{\pi}{}^{A_{2\vert s\vert}}$ if 
$s=-\vert s\vert\leq0$, and $\bar{\mathring{\pi}}{}^{A'_1}\cdots\bar
{\mathring{\pi}}{}^{A'_{2s}}$ if $s=\vert s\vert>0$. Then the next step is the 
generation of the representation space for the whole group $E(3)$ from this 
one dimensional space by the elements of $SU(2)$ that do not leave 
$\mathring{p}^i$ fixed, and by the translations in $E(3)$. Clearly, this is 
just the construction of the bundle of totally symmetric unprimed N-type 
spinors $\phi^{A_1...A_{2\vert s\vert}}$ on ${\cal S}$ if $s=-\vert s\vert$, and of 
the totally symmetric primed N-type spinors $\chi^{A'_1...A'_{2s}}$ on ${\cal S}$ 
if $s=\vert s\vert$. The resulting bundles are isomorphic to the line bundles 
${\cal O}(-2s)$ with the corresponding $s$. 

Next we determine the explicit form of the representation of $E(3)$ by 
operators ${\bf U}(\xi^i,A^A{}_B)$ acting on the spinor fields. The 
action of $SU(2)$ e.g. on the spinor field $\phi^{A_1...A_{2\vert s\vert}}$ is 
defined by $({\bf U}(0,A^A{}_B)\phi)^{A_1...A_{2\vert s\vert}}(p^i):=A^{A_1}{}_{B_1}
\cdots A^{A_{2\vert s\vert}}{}_{B_{2\vert s\vert}}\phi^{B_1...B_{2\vert s\vert}}(\Lambda(A^{-1})
^i{}_jp^j)$; while the action of the translation with $\xi^i$ is $({\bf U}
(\xi^i,\delta^A_B)\phi)^{A_1...A_{2\vert s\vert}}(p^i):=\exp({\rm i}p_i\xi^i/\hbar)
\phi^{A_1...A_{2\vert s\vert}}(p^i)$, the multiplication by the phase factor $\exp
({\rm i}p_i\xi^i/\hbar)$. These transformations provide a representation of 
$E(3)$, and these are unitary with respect to the scalar product 
(\ref{eq:A.1.5}). Moreover, since the functions $\exp({\rm i}p_i\xi^i/\hbar)$ 
on ${\cal S}$ span a dense subspace in the space of the square integrable 
functions on ${\cal S}$, this representation is irreducible, too. The $L_2$ 
space of the spinor fields with given $s$, i.e. the carrier space of the 
unitary, irreducible representation of $E(3)$ with given $P$ and $w$, will 
be denoted by ${\cal H}_s$. As we will see below, $w$ is linked to $s$, 
because we will find that $w=\hbar Ps$. 

In this representation, still for $s=-\vert s\vert\leq0$, the operators 
${\bf p}_i$ and ${\bf J}_{ij}:=\varepsilon_{ijk}{\bf J}^k$ are defined to be 
the densely defined self-adjoint generators of these transformations: If 
$\xi^i=uT^i$, then ${\bf p}_i$ is defined by $({\rm i}/\hbar)T^i{\bf p}_i
\phi^{A_1...}:=\frac{\rm d}{{\rm d}u}(({\bf U}(uT^i,\delta^A_B)\phi)
^{A_1...})\vert_{u=0}$. Here the limit in the definition of the derivative is 
meant in the strong topology of ${\cal H}_s$. Evaluating this, we find that 
${\bf p}_i$ is the multiplication operator. We already determined even the 
specific form of the self-adjoint generators ${\bf J}_{ij}$ of the action of 
$SU(2)$ in \cite{Sz1}, and we do not repeat that derivation here. We obtain 
\begin{eqnarray}
{\bf p}_i\phi^{A_1...A_{2\vert s\vert}}\!\!\!\!&=\!\!\!\!&p_i\phi^{A_1...A_{2\vert s\vert}},
  \label{eq:A.2.1a} \\
{\bf J}_{ij}\phi^{A_1...A_{2\vert s\vert}}\!\!\!\!&=\!\!\!\!&{\rm i}\hbar\Bigl(p_j
  \frac{\partial}{\partial p^i}-p_i\frac{\partial}{\partial p^j}\Bigr)\phi
  ^{A_1...A_{2\vert s\vert}}+\sqrt{2}\hbar\,s\,\varepsilon_{ij}{}^k\sigma_k
  ^{(A_1}{}_{(B_1}\delta^{A_2}_{B_2}\cdots\delta^{A_{2\vert s\vert})}
  _{B_{2\vert s\vert})}\phi^{B_1...B_{2\vert s\vert}}. \hskip 10pt \label{eq:A.2.1b}
\end{eqnarray}
Here $\sigma^A_i{}_B$ are the standard $SU(2)$ Pauli matrices (including the 
factor $1/\sqrt{2}$), which are just the unitary spinor form of the three 
non-trivial $SL(2,\mathbb{C})$ Pauli matrices, and, with the present sign 
convention, these are given by $\sigma^A_i{}_B=\delta_{ij}\varepsilon^{AC}
\sigma^j{}_{CB'}\sqrt{2}\sigma^{0B'}{}_B$. Since ${\bf p}_i$ is a 
multiplication operator and ${\cal S}$ is compact, it is well defined and 
is bounded on the whole of ${\cal H}_s$. On the other hand, ${\bf J}_{ij}$ is 
well defined only on the \emph{dense subspace} of the smooth spinor fields 
in ${\cal H}_s$. Then it is a straightforward calculation to check that these 
operators do, indeed, satisfy the defining commutation relations of $e(3)$ 
on the appropriate dense subspaces; and hence that (\ref{eq:A.2.1a}) and 
(\ref{eq:A.2.1b}) provide a unitary irreducible representation of the Lie 
algebra $e(3)$. 

By (\ref{eq:A.2.1a}) the Casimir operator ${\bf P}^2$ on ${\cal H}_s$ is 
simply the multiplication by $P^2$. In \cite{Sz1} we already calculated the 
action of $p^i{\bf J}_i$ on the spinor fields $\phi^{A_1...A_{2\vert s\vert}}$. 
However, by (\ref{eq:A.2.1a}) in the irreducible representation labeled 
by $P$ this is just the Casimir operator ${\bf W}$. Hence, 
\begin{equation}
{\bf W}\phi^{A_1...A_{2\vert s\vert}}={\bf p}^i{\bf J}_i\phi^{A_1...A_{2\vert s\vert}}=
\hbar Ps\phi^{A_1...A_{2\vert s\vert}}. \label{eq:A.2.2}
\end{equation}
Therefore, ${\bf W}$ is, in fact, proportional to the identity operator, 
and the value of ${\bf W}$ on ${\cal H}_s$ is $w=\hbar Ps$. Thus the 
representation can be labelled equally well by $P$ and $w$ or by $P$ and 
$s$. Repeating the previous analysis if $s=\vert s\vert>0$, we find the 
same expression for both ${\bf J}_{ij}$ and ${\bf W}$. Therefore, we 
obtained the unitary, irreducible representations of the Euclidean group 
$E(3)$ just in the geometrical form that we have for the Poincar\'e group, 
and summarized in \cite{StWi}. 

In \cite{Sz1}, we also calculated the contractions $m^i{\bf J}_i$ and $\bar 
m^i{\bf J}_i$. Using these, equation (\ref{eq:A.2.2}) and the fact that on 
$U_n$ the spinor field is $\phi^{A_1...A_{2\vert s\vert}}=\phi\pi^{A_1}\cdots\pi
^{A_{2\vert s\vert}}$, we obtained 
\begin{eqnarray*}
{\bf J}_i\phi^{A_1...A_{2\vert s\vert}}\!\!\!\!&=\!\!\!\!&\bigl(m_i\bar m^j
  +\bar m_im^j+\frac{1}{P^2}p_ip^j\bigr){\bf J}_j\phi^{A_1...A_{2\vert s\vert}} \\
\!\!\!\!&=\!\!\!\!&P\hbar\Bigl(m_i{\edth}'\phi-\bar m_i{\edth}\phi+s\frac{p_i}
  {P^2}\phi\Bigr)\pi^{A_1}\cdots\pi^{A_{2\vert s\vert}}.
\end{eqnarray*}
Therefore, defining the action of ${\bf J}_i$ on the spin weighted function 
$\phi$ with spin weight $s$ simply by ${\bf J}_i\phi:=(\sqrt{2}P)
^{-\vert s\vert}({\bf J}_i\phi_{A_1...A_{2\vert s\vert}})\iota^{A_1}\cdots\iota
^{A_{2\vert s\vert}}$, we obtain 
\begin{equation}
{\bf J}_i\phi=P\hbar\Bigl(m_i{\edth}'\phi-\bar m_i{\edth}\phi+s\frac{n_i}{P}
\phi\Bigr). \label{eq:A.2.3}
\end{equation}
The first two terms together on the right, denoted by ${\bf L}^i$, are the 
orbital, while the third is the spin part of the angular momentum. The former 
can also be written as $P^2\varepsilon_{ijk}{\bf L}^k={\bf C}_i{\bf p}_j-
{\bf C}_j{\bf p}_i$, from which the expression 
\begin{equation}
{\bf C}_i\phi={\rm i}\hbar\Bigl(P^2m_i\,{\edth}'\phi+P^2\bar m_i\,{\edth}
  \phi-p_i\phi\Bigr) \label{eq:A.2.4}
\end{equation}
for the centre-of-mass vector operator follows. 

By (\ref{eq:A.2.3}) it is easy to compute the only Casimir operator 
${\bf J}_i{\bf J}^i$ of the $su(2)$ subalgebra. It is given by 
\begin{equation}
{\bf J}_i{\bf J}^i\phi=\hbar^2\Bigl(-P^2\bigl({\edth}{\edth}'+{\edth}'{\edth}
\bigr)\phi+s^2\phi\Bigr); \label{eq:A.2.5}
\end{equation}
where ${\edth}{\edth}'+{\edth}'{\edth}$ is just the Laplace operator on 
${\cal S}$. Hence the spectrum of ${\bf J}_i{\bf J}^i$ is $j(j+1)\hbar^2$, 
$j=\vert s\vert,\vert s\vert+1,...$,  as we expected, and, for fixed $s$ 
and $j$, the $SU(2)$-irreducible subspaces in ${\cal H}_s$ are spanned by 
the spin weighted spherical harmonics ${}_sY_{j,m}$, $m=-j,-j+1,...,j$. 

If $P=0$, then ${\bf p}^i$ is represented by the identity operator. In this 
representation, the value $w$ of ${\bf W}$ is necessarily zero, but this 
representation is \emph{not} irreducible with respect to $SU(2)$. It is an 
infinite direct sum of unitary irreducible representations of $SU(2)$ 
labeled by the value $j$ of the Casimir operator of $su(2)$.


\end{document}